\documentclass[pra,aps,reprint,superscriptaddress]{revtex4-2} 
\newif\ifarxiv 
\usepackage{natbib}
\usepackage{graphicx}
\usepackage{float}
\usepackage{amsfonts,amstext}
\usepackage{url}
\usepackage{xcolor}
\usepackage{braket}
\usepackage{soul}
\usepackage{amsthm}
\usepackage{physics}
\usepackage{hhline}
\usepackage{bbm}
\usepackage{enumerate}
\usepackage{comment}
\usepackage{colortbl}
\usepackage{amssymb}
\usepackage{comment}
\usepackage{float}
\usepackage{placeins}
\usepackage[colorlinks=true,linkcolor=blue,citecolor=blue,urlcolor=blue]{hyperref}
\usepackage{cleveref}

\renewcommand{\emph}[1]{\textit{#1}}

\newtheoremstyle{theorem}
	{6pt}
	{}
	{\itshape}
	{}
	{\bfseries}
	{:}
	{.5em}
	{}
\theoremstyle{theorem}

\newtheorem{lem}{Lemma}[section]

\newtheorem{prop}{Proposition}[section]

\begin{document}

\title{Multi-copy and Catalytic Superactivation of Genuine Multipartite Nonlocality}

\author{Bora Ulu}
\affiliation{Department of Applied Physics, University of Geneva, Switzerland}

\author{Mirjam Weilenmann}
\affiliation{Inria, Télécom Paris - LTCI, Institut Polytechnique de Paris, France}

\author{Nicolas Brunner}
\affiliation{Department of Applied Physics, University of Geneva, Switzerland}

\date{\today}

\begin{abstract}
Genuine multipartite Bell nonlocality (GMNL) aims at capturing correlations between many distant observers that are globally Bell nonlocal. A commonly used definition for GMNL is based on biseparable models constructed from non-signalling correlations. Here we uncover an effect of superactivation of GMNL in this framework. Specifically, by locally wiring together two copies of a biseparable correlation, i.e. not GMNL, one can obtain another correlation that is GMNL. We show that this is possible for any number of parties, and for quantum-realizable correlations. Finally, we show that superactivation of GMNL is also possible at the single-copy level, via a catalytic protocol involving only local wirings. These results question the operational meaning of this definition of GMNL. 
   
\end{abstract}

\maketitle

\section{Introduction}

Quantum physics allows for correlations that are incompatible with a natural notion of locality, as formalized by Bell. These phenomena are particularly interesting in the multipartite case, where a number of distant observers (three or more) perform local measurements on a shared entangled state. Similary to entanglement, there exist many different notions of multipartite Bell nonlocality, the strongest of which is genuine multipartite nonlocality (GMNL). Beyond the conceptual interest, these ideas have played a significant role in quantum information, notably in quantum communication, see e.g. \cite{Scarani2001,Ribeiro2018,Holz2020,Grasselli_2023}, as well as for the device-independent detection of genuine multipartite entanglement \cite{bancal2011_diew}.

The concept of GMNL was first introduced by Svetlichny~\cite{svetlichny87}, for capturing multipartite quantum correlations that are globally (and irreducibly) Bell nonlocal. Specifically, these correlations cannot be accounted for by any biseparable local model, where the parties can join into two groups. Various families of Bell inequalities have been introduced for certifying GMNL \cite{Seevinck2002,Collins2002,Bancal2011}, some of which have been tested experimentally \cite{Lavoie2009}.

In turn, it was realized that Svetlichny's original definition of GMNL leads to some inconsistencies when considering multipartite nonlocality from an operational (or resource theoretic) perspective \cite{Gallego2012}. In a nutshell, in Svetlichny's definition of biseparable models, parties grouping together can produce arbitrary distributions, featuring signalling in multiple directions at once, which can lead to grand-father type paradoxes \cite{Bancal2013}. 

This motivated altenative (and refined) defintions of GMNL, where signaling properties in biseparable models are limited (e.g. imposed by a time ordering between the parties) \cite{Gallego2012,Bancal2013}, which avoids the previously mentioned paradoxes. An even more conservative definition was proposed assuming that biseparable models involve only non-signaling distributions~\cite{Bancal2013}. This last definition appears to be by now largely accepted as a consistent definition of GMNL and is commonly used in the literature, see e.g. \cite{Baccari2019,Horodecki2019,Curchod_2019,Contreras_Tejada_2021}.

In this work, we revisit these ideas, and present results that question the operational meaningfullness of the currently used definition of GMNL. Specifically we exhibit a phenomenon of superactivation of GMNL; by combining several copies of biseparable distributions, it is possible to obtain a new distribution that is GMNL. Importantly, the combination of the copies is performed exclusively via local and classical operations, namely local wirings. This implies that the definition of biseparability is not stable ``under tensorisation'' (i.e. taking several copies), which would be a desirable property from an operational/resource-theoretic point of view. Notably, this is in stark contrast to the standard notion of Bell locality, which is tensor stable under local operations, both in the bipartite and multipartite case; combining any number of copies of local distributions via local wirings always leads to a final distribution that remains local.

After presenting the relevant notions and notations, we first provide a simple and illustrative example of superactivation of GMNL via local wirings, considering nonlocal boxes in a tripartite setting. While this initial example involves distributions that are not realizable in quantum theory (but are nevertheless non-signaling), we show that superactivation of GMNL via local wirings is also possible for quantum distributions, constructing explicit examples for any number of parties $N \geq 3$. Finally, we show that superactivation of GMNL is also possible in the single-copy regime, considering a catalytic scenario. Here an initial biseparable distribution is transformed with the help of a shared catalyst (in the form of another distribution) and local operations, into a distribution that is GMNL.

Our work complements previous ones demonstrating superactivation of nonlocality at the level of quantum states \cite{palazuelos2012,Cavalcanti2013,Quintino2016,Contreras_Tejada_2021,miethlinger2026superactivationgenuinemultipartitebell}; starting from some entangled state admitting a local hidden variable model, Bell inequality violation can be obtained by performing joint measurements on many copies of the state. In contrast, our work operates at the level of probability distributions. In this sense our results concern the model-independent definition of GMNL. As mentioned above, the superactivation effect shows that the current definition of GMNL is not fully satisfactory from an operational perspective. In the conclusion we discuss a possible alternative definition, recently developed in a series of works \cite{Navascues2020,kraft2021,CoiteuxRoy2021,Coiteux2021pra}, inspired by the field of network Bell nonlocality \cite{Fritz2012, Tavakoli2022}.

\section{Preliminaries}

\subsection{Genuine Multipartite Non-locality} \label{sec:preliminaries}

\begin{figure}[t!]
    \centering
    \includegraphics[width=1.0\columnwidth]{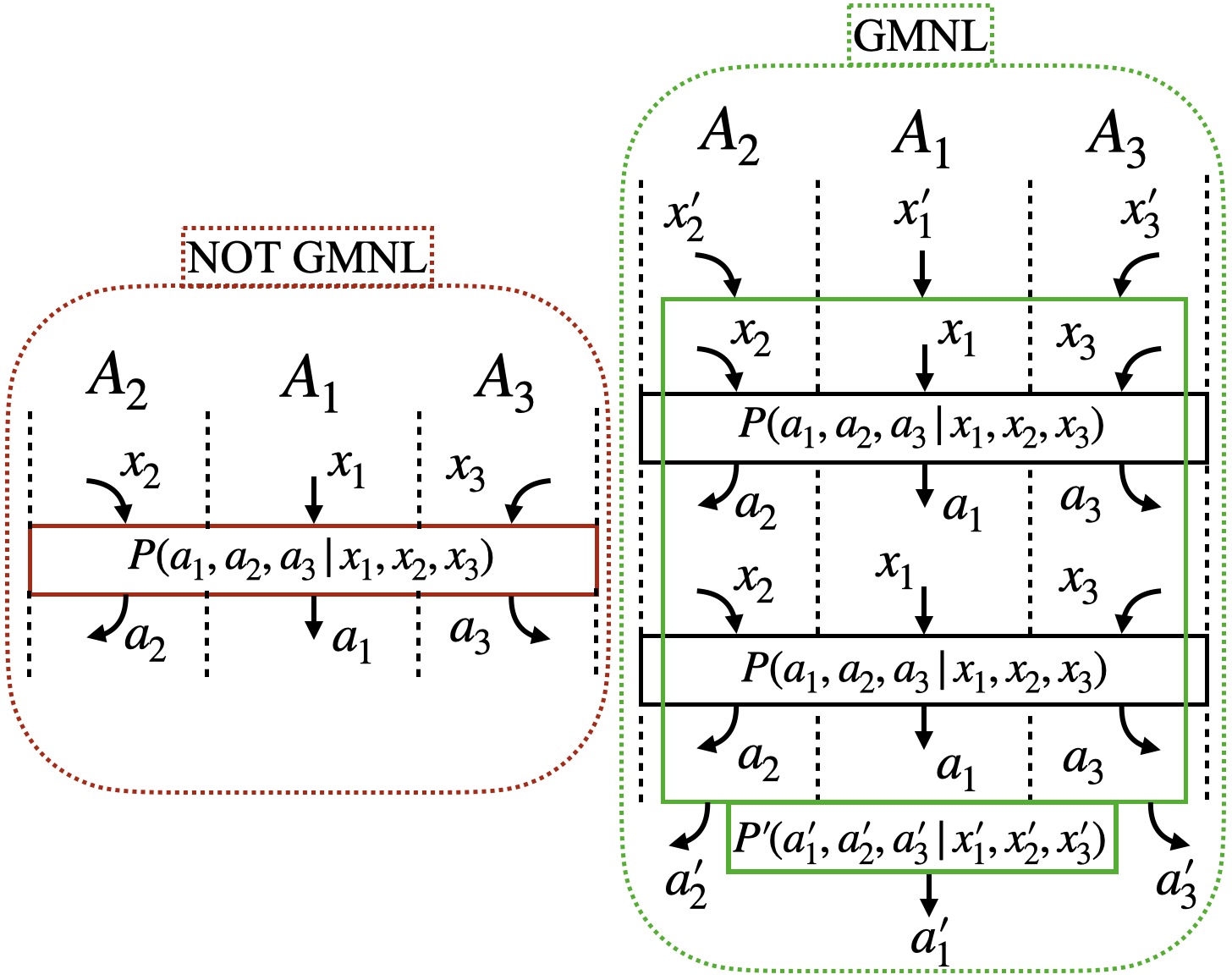}
    \caption{Illustrative example of multi-copy superactivation of GMNL. The three parties $A_1,A_2$ and $A_3$ begin with a box $P$ with local inputs $x_1,x_2,x_3$ and outputs $a_1,a_2,a_3$ that is biseparable (not GMNL). After wiring two copies of this box $P$, i.e. using local classical operations on their local inputs and outputs, they obtain another box $P'$ that is GMNL. }
    \label{fig1}
\end{figure}

Consider a Bell scenario with $K$ parties $\{A_k\}_{k=1}^{K}$ and with the resulting conditional distribution $P(a_1,...,a_K|x_1,...,x_K)$ where for $i\in\{1,2,...,K\}$, the inputs of the parties are given by $x_i\in \mathcal{X}_i$ and similarly the outputs are given by $a_i\in \mathcal{A}_i$ \footnote{Here, $ \mathcal{X}_i$ and $\mathcal{A}_i$ are considered to be finite sets.}. Then, if

\begin{align}\nonumber
    P(a_1,...,a_K|x_1,...,x_K) = \sum_{\lambda} q(\lambda) \prod_{k=1}^K P(a_k|x_k,\lambda)
\end{align}
where $0\leq q(\lambda)\leq1$ and $\sum_\lambda q(\lambda) = 1$, we say that $P$ belongs to the set of $K$-partite local distributions $\mathcal{L}_K$, otherwise it is $K$-party non-local. Local distributions can be expressed as convex combinations of local deterministic distributions. For a given  local deterministic distribution $P(\vec a|\vec x)$ where $\vec a = (a_1,...,a_K)$ and $\vec x =(x_1,...,x_K)$ there always exist local response functions $f_k:\mathcal{X}_i \rightarrow \mathcal{A}_i$ from which we can obtain $P(\vec a |\vec x)=\prod_{k=1}^K\delta_{a_k,f_k(x_k)}$. For brevity, we can express $P(\vec a |\vec x)=\delta_{\vec a,f(\vec x)}$ where $f(\vec x) = (f_1(x_1),...,f_K(x_K))$. 

We further define the set $\mathcal{NS}_{K}$ of $K$-partite non-signalling correlations: Given the set
\begin{equation}\nonumber
\mathcal{B} = \{\alpha \subset \{1,...,K\}:1\in\alpha, \alpha \neq \{1,...,K\}\}
\end{equation}
a $K$-partite distribution $P$ is non-signalling if it satisfies for all $\alpha\in \mathcal{B}$ and its complement $\bar \alpha = \{1,...,K\}\setminus \alpha$,
\begin{equation}\nonumber
    \sum_{\vec{a}_\alpha}P(\vec{a}_\alpha,\vec{a}_{\bar \alpha}|\vec{x}_\alpha,\vec{x}_{\bar \alpha}) = \sum_{\vec{a}_\alpha}P(\vec{a}_\alpha,\vec{a}_{\bar \alpha}|\vec{x}_\alpha^\prime,\vec{x}_{\bar \alpha})
\end{equation}
where $\vec{x}_{ \alpha} = (x_i)_{i\in \alpha}$, $\vec{x}'_{ \alpha} = (x'_i)_{i\in \alpha}$ and $\vec{a}_{ \alpha} = (a_i)_{i\in \alpha}$. This captures the idea that the inputs of any subset of parties cannot influence the outputs of the remaining parties, which compatibility with special relativity requires if the $K$ parties are space like separated. Notice that quantum correlations have this property, i.e., the sets of $K$-partite quantum correlations are a subset of $\mathcal{NS}_{K}$.

Let us know move to the concept of GMNL, first introduced by Svetlichny~\cite{svetlichny87}. Our focus is on a refined definition proposed in \cite{Bancal2013}, which enforces the non-signaling constraint among all subsets of parties. Formally, we first need to introduce the subset of correlations that are $K$-partite biseparable no-signalling, $\mathcal{N}_{K}^2$,  which is the set of $P$ that allow for a decomposition
\begin{align}\label{gmnl_set}\nonumber
    &P(a_1,...,a_K|x_1,...,x_K) =\\& \sum_{\alpha\in\mathcal{B}}\sum_{\lambda} q_\alpha(\lambda)P^{|\alpha|}_{NS}(\vec{a}_{\alpha} |\vec{x}_{ \alpha},\lambda)P_{NS}^{K-|\alpha|}(\vec a_{\bar \alpha}|\vec x_{\bar \alpha},\lambda)
\end{align}
over the hidden variable $\lambda$ where $0\leq q_\alpha(\lambda)\leq 1$ as well as $\sum_{\alpha\in\mathcal{B}}\sum_{\lambda}q_\alpha(\lambda)=1$ and $P^{k}_{NS} \in \mathcal{NS}_k \forall k$ can be any non-signalling distribution. Notice that $\mathcal{L}_K\subset\mathcal{N}^2_{K}  \subset \mathcal{NS}_{K}$. Similarly, we define the more restrictive set of $K$-partite $(K-1)$-separable no-signalling distributions $\mathcal{N}_K^{K-1}$, which is the set of $P$ that allows for a decomposition
\begin{align}\nonumber
    &P(a_1,...,a_K|x_1,...,x_K) = \\&\sum_{\alpha\in\mathcal{B}:|\alpha|=2}\sum_{\lambda}q_{\alpha}(\lambda)P^2_{NS}(\vec a_\alpha|\vec x_\alpha,\lambda)\delta_{\vec a_{\bar \alpha},f^\lambda_{\bar \alpha}(\vec x_{\bar \alpha})}
\end{align}
where $f^\lambda_{\bar \alpha}$ are local response functions for all $\lambda$ and $\bar \alpha$ and $\sum_{\alpha\in\mathcal{B}:|\alpha|=2}\sum_{\lambda}q_\alpha(\lambda)=1$. Any $(K-1)$-separable distribution is by definition also biseparable and thus
\begin{align}
\mathcal{L}_K\subset\mathcal{N}_K^{K-1}\subset \mathcal{N}^2_{K}  \subset \mathcal{NS}_{K}.
\end{align}
For a distribution to be $K$-party GMNL,
 the distribution $P \in \mathcal{NS}_{K}$ must be nonlocal and must not accept any convex decomposition over bipartitions of $(k<K)$-partite nonlocal distributions, i.e., $P \notin \mathcal{N}^2_K$. Note that, similarly to $\mathcal{L}_K$, the sets $\mathcal{N}^2_K$ and $\mathcal{N}_K^{K-1}$ are also polytopes. The facets of theses polytope are linear Bell-type inequalities. A violation of one of the inequalities (facets) of $\mathcal{N}^2_K$ implies the presence of GMNL. For the simplest case of $K=3$ and binary inputs and outputs, all such GMNL inequalities have been derived in \cite{Bancal2013} using linear programming techniques, and some of these inequalities can be violated via quantum distributions.

\subsection{Wirings}

In this work, we are interested in a scenario where the parties have access to several copies of some input/output ``nonlocal box'', represented by a conditional distribution $P(a_1,...,a_k|x_1,...,x_k)$. We will say that a nonlocal box is GNML or $(K-1)$-separable if the distribution that is representing it has the respective property.  Each of the $K$ parties is then allowed to classically process their local inputs and outcomes, while any form of communication with other parties is prohibited.

The operations that each party can perform on these boxes are then naturally captured by means of (classical) wirings, where e.g. the output of the first box can be used to determine the input to the second box and so on. Wirings arguably represent the minimal class of local operations that must be considered in any operational/resource-theoretic approach to Bell nonlocality \cite{Barrett_2005}. Interestingly, wirings can be used to increase the amount of nonlocality of a given nonlocal box (by wiring several copies of the box), a process termed nonlocality distillation \cite{Forster_2009,Brunner_2009,Hoyer_2010, Brito_2019,Eftaxias_2023,Naik_2023,Hoyer_2026}, and for activating key in device-independent key distribution \cite{Ulu_2025}.

In this work, we will consider a subset of wirings called ordered wirings. 
Consider two boxes $P_1(a_{(1)1},...,a_{(1)K}|x_{(1)1},...,x_{(1)K})$ and $P_2(a_{(2)1},...,a_{(2)K}|x_{(2)1},...,x_{(2)K})$ such that each party $A_k$ has access to the classical bits $(a_{(1)k},a_{(2)k},x_{(1)k},x_{(2)k})$. In an ordered wiring, for each input $x_k$ the party $A_k$ chooses an input for their first box $x_{(1)k}$ (which may depend on $x_k$). They then choose the input for their second box $x_{(2)k}$ which may depend on $(a_{(1)k},x_{(1)k},x_k)$. Finally, based on  $(a_{(1)k},a_{(2)k},x_{(1)k},x_{(2)k},x_k)$ they choose the final outcome $a_k$. All parties follow a similar procedure on their parts of the boxes and together this leads to a new box $P^{\prime}(a_1,...,a_K|x_1,...,x_K)$. This local processing implemented by the parties can be represented as a map $\mathcal{F}$, which takes the two initial boxes $P_1$ and $P_2$ to the final nonlocal box $P'=\mathcal{F}(P_1,P_2)$. 

Here we are mostly interested in the case where the two initial boxes are copies of the same box $P=P_1=P_2$. In particular we show that there exist boxes that are biseparable, yet the final wired box $P^\prime$ becomes GMNL, i.e. is provably no longer biseparable. This effect is termed superactivation of GMNL.

\section{Illustrative Example of Superactivation of GMNL}

In this article, we are interested in whether we can start from resources without GNML and create it by simple local operations. 
 Specifically, if $K$-parties start with resources that are not GMNL and after local operations on multiple copies of the same resource can they obtain another resource that is GMNL? We call this super-activation of GMNL. We will consider super-activation at the level of correlations or boxes, meaning that we make no assumptions on the implementation of quantum states and measurements that generate it. Before focusing on the quantum examples later, we demonstrate in this section  the existence of multi-copy GMNL superactivation by means of a simpler, post-quantum example.

Let us consider the tripartite Bell scenario ($K=3$) with 2 inputs and 2 outputs for each party. We construct a tripartite distribution that by combining a bipartite nonlocal box with a single-party deterministic box. For the bipartite distribution we consider the Popescu-Rohrlich distribution (so-called PR-box)~\cite{tsirelson, PRbox}, which is an extremal point of the bipartite non-signalling polytope. It is characterized by
\begin{align}\label{PR}
    P_{\textrm{PR}}(a_1,a_2|x_1,x_2) = \begin{cases}
1/2, & \text{ if } a_1\oplus a_2=x_1x_2.\\
0, & \text{ otherwise}
\end{cases}
\end{align}
For the deterministic single-party distribution, we take $P_{\mathrm{det}}(a_k|x_k)=\delta_{0,a_k}$ that outputs $0$. Combining these objects, we construct the following tripartite distribution
\begin{align}\nonumber
    P(a_1,a_2,a_3|x_1,x_2,x_3)& = \frac{1}{2}P_{\textrm{PR}}(a_1,a_2|x_1,x_2)\delta_{0,a_3} \\&+ \frac{1}{2}P_{\textrm{PR}}(a_1,a_3|x_1,x_3)\delta_{0,a_2}.\label{simple_ex}
\end{align}
Clearly, this distribution is biseparable, by construction and belongs to $\mathcal{N}^2_3$.

Now, consider a scenario where two copies of the above distribution, i.e. $P(a_{(1)1},a_{(1)2},a_{(1)3}|x_{(1)1},x_{(1)2},x_{(1)3})$ and $P(a_{(2)1},a_{(2)2},a_{(2)3}|x_{(2)1},x_{(2)2},x_{(2)3})$ are wired together using the wiring $\mathcal{V}[P,P]=P'(a_1',a_2',a_3'|x_1',x_2',x_3')$ defined by the relations
\begin{align}
    x_{(1)k}=x_{(2)k}=x_k',\quad a_k' = a_{(1)k}\oplus a_{(2)k} \label{wiring}
\end{align}
for all $k\in\{1,2,3\}$. Then, the distribution after this wiring becomes
\begin{align}
    P'(a_1',a_2',a_3'|x_1',x_2',x_3') &= \frac{1}{4}\mathcal{V}[P_1,P_1] + \frac{1}{4}\mathcal{V}[P_2,P_2] \nonumber \\&+ \frac{1}{4}\mathcal{V}[P_1,P_2] + \frac{1}{4}\mathcal{V}[P_2,P_1]
\end{align}
where $P_1=P_{\textrm{PR}}(a_1,a_2|x_1,x_2)\delta_{0,a_3}$ and $P_2 = P_{\textrm{PR}}(a_1,a_3|x_1,x_3)\delta_{0,a_2}$. Each of these elements appear as follows: For $\mathcal{V}[P_1,P_1]$ the wiring states $a_1' = a_{(1)1}\oplus a_{(2)1}$ and $a_2' = a_{(1)2}\oplus a_{(2)2}$ therefore, $a_1'\oplus a_2' = (a_{(1)1}\oplus a_{(1)2})\oplus (a_{(2)1}\oplus a_{(2)2})$. From the distribution $P_1$ we know that $a_{(1)1}\oplus a_{(1)2} = x_{(1)1}x_{(1)2} = x_1'x_2'$ and $a_{(2)1}\oplus a_{(2)2} =  x_1'x_2'$ and hence $a_1'\oplus a_2'=0$. In addition, $a_3'=0$, which implies that $\mathcal{V}[P_1,P_1] = P_C(a_1',a_2'|x_1',x_2')\delta_{0,a_3'}$ where $P_C(a_1',a_2'|x_1',x_2') = \delta_{a_1',a_2'}/2$ is the perfectly correlated distribution for all inputs and outputs. A similar procedure reveals that $\mathcal{V}[P_2,P_2]=P_C(a_1',a_3'|x_1',x_3') \delta_{0,a_2'}$. 

\begin{figure}[t!]
    \centering
    \includegraphics[width=1.0\columnwidth]{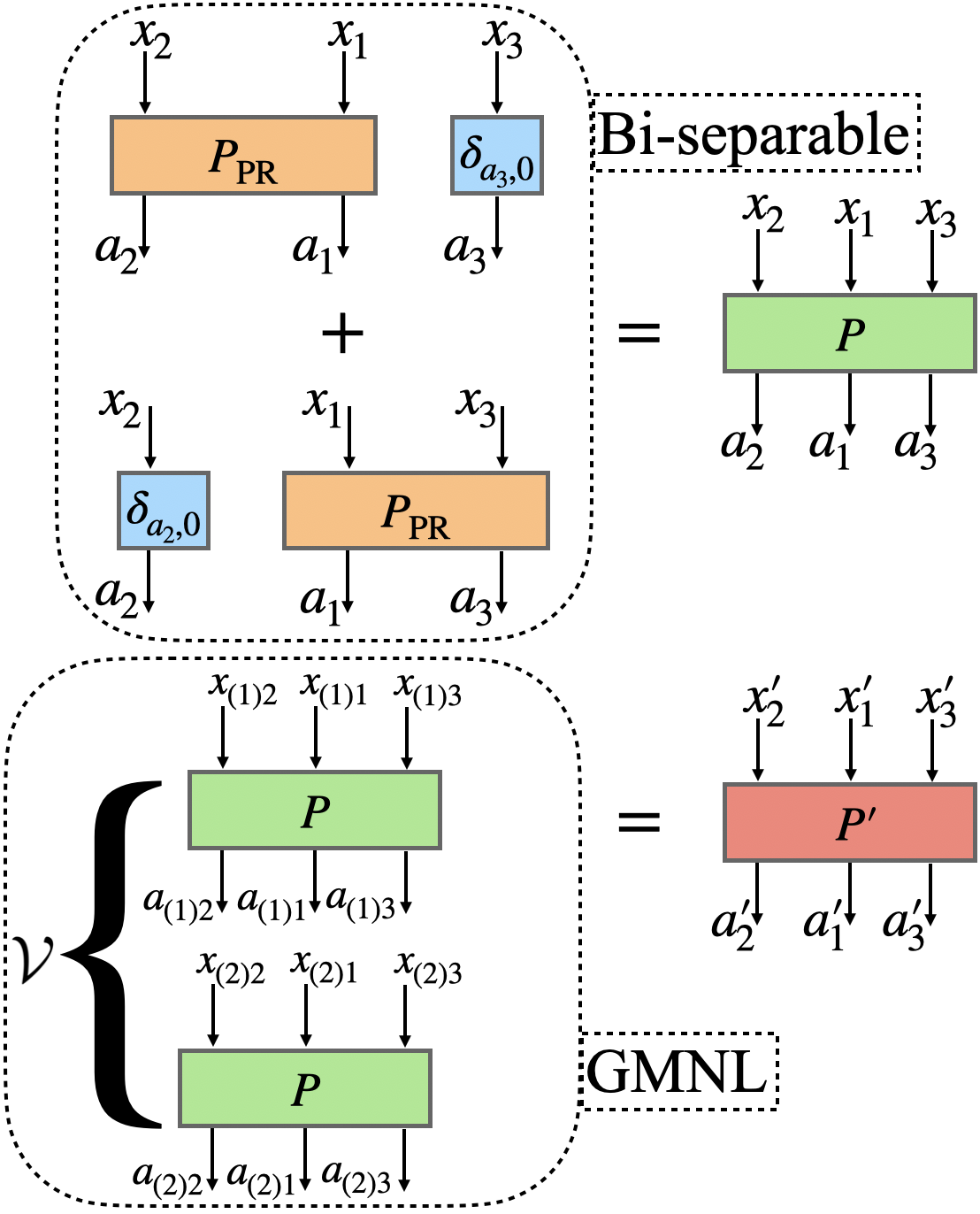}
    \caption{Illustrative example of superactivation of GMNL. In the top left we have the biseparable mixture $P$ as in \eqref{simple_ex} and GMNL is activated using the wiring $\mathcal{V}[P,P]=P'$ defined in \eqref{wiring}.}
    \label{fig5}
\end{figure}

Now we look at the cross terms, first notice that $\mathcal{V}[P_1,P_2]=\mathcal{V}[P_2,P_1]$ due to the symmetry of the wiring under exchange of the rounds. For $\mathcal{V}[P_1,P_2]$ we have $a_2' = a_{(1)2}\oplus a_{(2)2} = a_{(1)2}$ and similarly $a_3' = a_{(2)3}$. Furthermore, $a_1' = a_{(1)1}\oplus a_{(2)1}$ and using  \eqref{PR}, which together with \eqref{wiring} implies $a_{(1)1} = x_1'x_2' \oplus a_{(1)2}$ and $a_{(2)1} = x_1'x_3' \oplus a_{(2)3}$, we get $a'_1 = x_1'x_2'\oplus x_1'x_3' \oplus a_2'\oplus a_3'$. Therefore, we find that $\mathcal{V}[P_1,P_2] = P_{\mathrm{GMNL}}(a_1',a_2',a_3'|x_1',x_2',x_3')$ where
\begin{align}\nonumber
    P_{\textrm{GMNL}}(a_1',&a_2',a_3'|x_1',x_2',x_3') \\&= \begin{cases}
1/4, & \text{ if } a_1'\oplus a_2'\oplus a_3'=x_1'(x_2'\oplus x_3'),\\
0, & \text{ otherwise}
\end{cases} \, .
\end{align}
This is an extremal non-signalling tripartite box that is not biseparable \cite{Barrett_2005,Pironio_2011}. 

To show that the overall distribution 
\begin{align}\nonumber
    P'(a_1',a_2',a_3'|x_1',x_2',x_3') &= \frac{1}{4}P_C(a_1',a_2'|x_1',x_2')\delta_{0,a_3'}\\&\nonumber + \frac{1}{4}P_C(a_1',a_3'|x_1',x_3')\delta_{0,a_2'}
    \\& + \frac{1}{2}P_{\mathrm{GMNL}}(a_1',a_2',a_3'|x_1',x_2',x_3').\label{Eq: finalbox}
\end{align}
is GMNL, we use an inequality derived in Ref. \cite{Bancal2013}, which is a facet of the polytope $\mathcal{N}^2_3$ for binary inputs and outcomes, hence satisfied by all biseparable distributions. Specifically, we use a relabeling of the inequality of class 169, denoted $I_{\mathrm{GMNL}}^{169}(P) \leq 0$; for the exact form of the inequality, see Appendix-\ref{app:explicit_ineq}.
The first two distributions appearing in the convex mixture saturate the biseparable bound, i.e. $I_{\mathrm{GMNL}}^{169}(P_C\delta)=0$, while $P_{\mathrm{GMNL}}$ violates it by $I_{\mathrm{GMNL}}^{169}(P_{\mathrm{GMNL}})=12$. Therefore, altogether we get 
\begin{equation}
I_{\mathrm{GMNL}}^{169}(P')=6>0
\end{equation}
which demonstrates that the wired box $P'$ in Eq. \eqref{Eq: finalbox} is GMNL. We conclude that GMNL can indeed be activated in the many-copy scenario: by writing two copies of a biseparable non-signalling box $P$, we obtain a wired box $P'$ that violates a witness for GMNL. 

From the above example, it is natural to ask a number of questions. First, one may wonder whether GMNL superactivation is possible for distributions that admit a quantum model? Second, is GMNL possible for distributions with any number of parties? In the following section, we answer both questions in the affirmative, by constructing a family of $N$-party quantum distributions that are biseparable, yet become GMNL when several copies are combined via a judiciously chosen wiring.

Note that we have considered adapting the above example of GMNL superactivation to the quantum case, specifically, by mixing the PR-box in distribution \eqref{simple_ex} with a maximally mixed local box (the uniform mixture of all deterministic local boxes) to a point where it admits a quantum realization. However, we could not find any instance of GMNL superactivation in this case. A completely different construction is presented in the next section.

\section{Multi-copy Superactivation of GMNL from quantum boxes}

In this section, we show that superactivation of GMNL is possible for any number of parties, starting from nonlocal boxes that admit a quantum realisation. Consider a scenario involving $K$ distant parties. We construct a family of quantum boxes $P$ which are $(K-1)$-separable, i.e. $P \in \mathcal{N}^{(K-1)}_K$. In turn, we show that by wiring two copies of $P$, we obtain a quantum box $P'$ which is GMNL, i.e. $P' \notin \mathcal{N}^2_K$.

To construct these examples, we reuse the structure of the illustrative example of the previous section. That is, we start from an elementary bipartite nonlocal quantum box, distributed between parties $A_1$ and $A_n$ (with $n=2,...,K$), while the remaining $K-2$ parties have a local deterministic box. In turn, we construct our $K$-party distribution $P$ by symmetrizing the above elementary distribution, by taking the convex mixture over the index $n$. By construction, this box is clearly $(K-1)$-separable, as each term in the mixture features a nonlocal box between only two parties. Finally, we show that by combining two copies of $P$ via an appropriate wiring, we obtain a final box $P'$ which is GMNL. To demonstrate this last statement, we construct a Bell inequality for witnessing GMNL based on the ideas of Ref. \cite{Contreras_Tejada_2021}.

Let us start by first discussing the building block in our construction, namely a family of bipartite quantum distributions arising from well-known Hardy paradox \cite{PhysRevLett.71.1665}. We term these Hardy correlations:
\begin{equation}\label{Eq: Born}
    P_H(a_1a_2|x_1x_2) = \text{Tr}\left[\ketbra{\psi}{\psi} \left(M_{A_1}^{x_1,a_1}\otimes M_{A_2}^{x_2,a_2}\right)\right], 
\end{equation}
where
\begin{equation}
    \ket{\psi} = \frac{\ket{00}+t\ket{11}}{\sqrt{1+t^2}}
\end{equation}
and the measurements 
\begin{align}\nonumber
    M_{A_1}^{x_1,a_1} &= \left[\mathbb{I}+(-1)^{a_1}O^{(x_1)}_1\right]/2\\\nonumber
    M_{A_2}^{x_2,a_2} &= \left[\mathbb{I}+(-1)^{a_2}O^{(x_2)}_2\right]/2
\end{align}
are defined in terms of the
 observables $O^{(x_1)}_1 = \cos(\xi_{x_1})\sigma_z + \sin(\xi_{x_1})\sigma_x$ and $O^{(x_2)}_2 = \cos(\gamma_{x_2})\sigma_z -\sin(\gamma_{x_2})\sigma_x$ where $\sigma_z$ and $\sigma_x$ are the Pauli operators 
\begin{align}\nonumber
    \xi_1 = \gamma_1 = -\arccos\left(\frac{t-1}{t+1}\right),\xi_0 = \gamma_0 = \arccos\left(\frac{t^3-1}{t^3+1}\right).
\end{align} 
The property of the Hardy correlations that will be crucial for this work is that they satisfy $P_H(01|01)=P_H(10|10)=P_H(00|11)=0$ and $P_H(00|00)=p(t)=t^2(1-t^2)^2/(1+t^2)(1+t^3)^2$. Hence, their nonlocality can be witnessed by the violation of the Clauser-Horne (CH) inequality \footnote{Note that the CH inequality is equivalent to the CHSH inequality as $I(P)=(I_{CHSH}(P)-2)/4$ \cite{Brunner_2014}.}
\begin{equation}\nonumber
    I(P)=P(00|00)-P(01|01)-P(10|10)-P(00|11)\leq 0
\end{equation} 
with $I(P_H)=p(t)>0$.

We now move to our scenario of interest, namely a Bell scenario with $K$ parties $\{A_k\}_{k=1}^{K}$ and with the resulting conditional distribution $P(a_1,...,a_K|x_1,...,x_K)$ where for $i\in\{1,2,...,K\}$, the inputs of the parties are given by $x_i\in \mathcal{X}_i$ and similarly the outputs are given by $a_i\in \mathcal{A}_i$. Let one of the parties $A_1$ have $2^{K-1}$ inputs and outputs and let these be expressed as bit strings of the form $x_1=(x_1^2,x_1^3,...,x_1^{K})$ and $a_1=(a_1^2,a_1^3,...,a_1^{K})$ such that $\{0,1\}^{K-1}\subseteq \mathcal{X}_1$ and $\{0,1\}^{K-1}\subseteq \mathcal{A}_1$. Let the rest of the parties $\{A_k\}_{k=2}^{K}$ have inputs and outputs $ \mathcal{X}_k=\{0,1\}$ and $\mathcal{A}_k=\{0,1\} $. Then we can define the following bipartite quantities between $A_1$ and $A_k$ as
\begin{widetext}
\begin{align}
    Q^P_k(u,v\mid p,q)
=
\sum_{\substack{a_1\in\{0,1\}^{K-1}\\ :a_1^k=u}}
P\Big(
a_1, a_k=v, \vec{a}_{\not {k}}=0
 \Big|
x_1=pe_k, x_k=q, \vec{x}_{\not{k}}=0 \Big). \label{marginals}
\end{align}
\end{widetext}
where $\vec{a}_{\not{k}}=(a_2,...,a_{k-1},a_{k+1},...,a_{K})$, $\vec{x}_{\not{k}}=(x_2,...,x_{k-1},x_{k+1},...,x_{K})$ and $e_k\in\{0,1\}^{K-1}$ is the unit bit string that is all zeros except for the $(k-1)$th element. Further, to these quantitites we can apply the CH functional 
\begin{equation}\label{CH_func}
    I(Q^P_k) = Q^P_k(00|00) - Q^P_k(01|01) - Q^P_k(10|10) - Q^P_k(00|11).
\end{equation}

In this work, we rely on the following  $K$-party GMNL inequality from~\cite{Contreras_Tejada_2021}, 
\begin{align}\nonumber
I_K(P)
=&
\sum_{k=2}^{K} I(Q^P_k)
+
P\left(0^{K-1},0,\ldots,0 \middle|\, 0^{K-1},0,\ldots,0\right)
\\& \label{top_ineq} -
\sum_{k=2}^{K} Q^P_k(0,0\mid 0,0)\leq 0 
\end{align}
such that $I_K(P)>0$ certifies that the box $P$ is $K$-party GMNL. For completeness, we provide a proof of this inequality in Appendix~\ref{app:lifting}.

Next, we show that GNML can indeed be super-activated. Specifically, we show that in a $K$-partite  Bell scenario for any $K>2$, starting with multiple copies of a biseparable quantum box that is not GMNL we can apply an ordered wiring of the inputs and outputs to obtain another box in the same scenario that violates the $K$-party GMNL witness in \eqref{top_ineq}, i.e. that is GMNL. This demonstrates that correlations themselves can be manipulated with local deterministic operations to super-activate GMNL. Note that, since we remain in the same Bell scenario, these copies can be regarded as consecutive rounds of the same experiment and thus are easy to implement. This means that rather than needing a genuinely multi-partite process or some complicated bipartite measurement, an experiment that can produce Hardy correlations between pairs of parties (combined with classical processing) is sufficient for establishing GMNL.

\begin{figure}[h!]
    \centering
    \includegraphics[width=1.0\columnwidth]{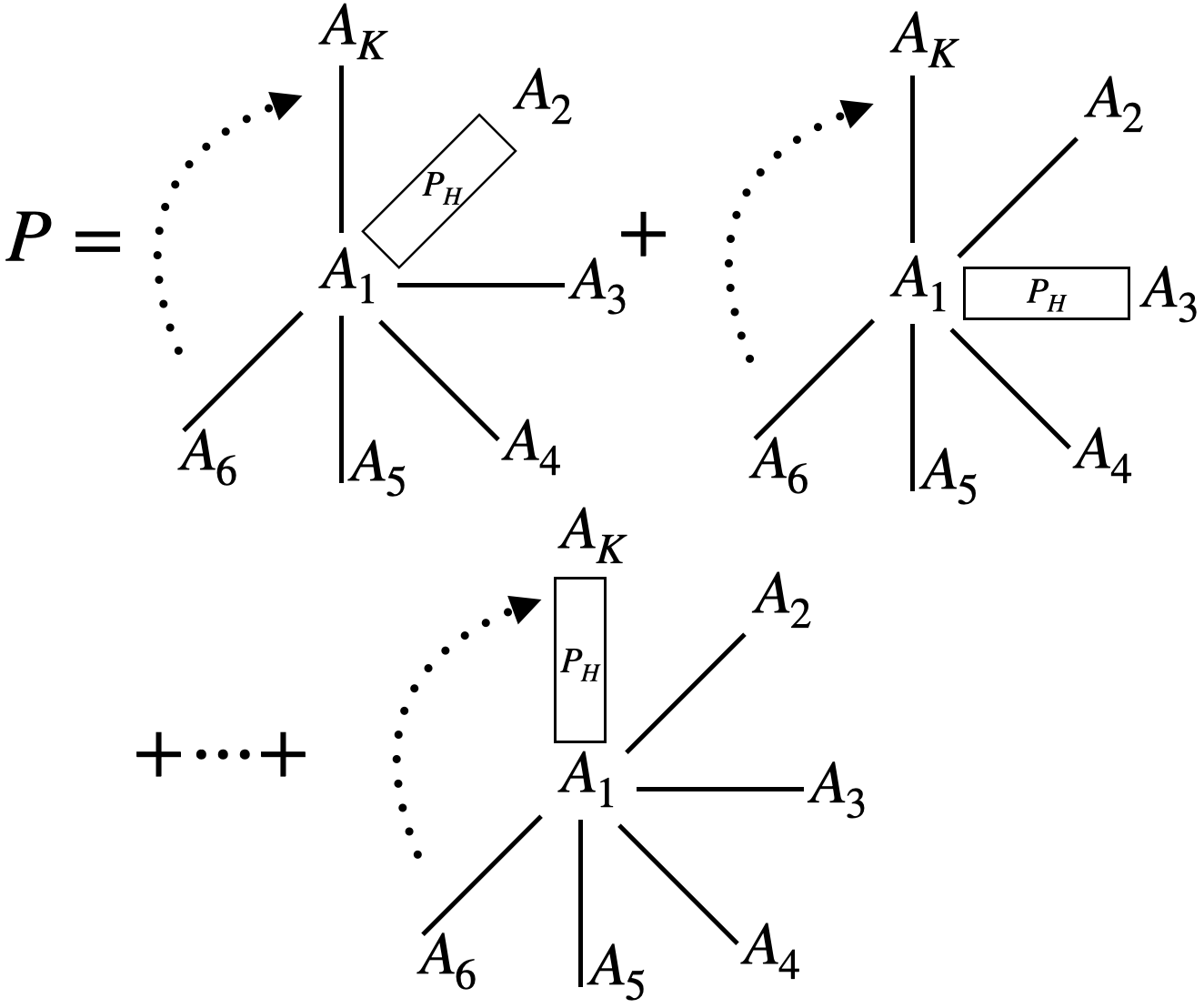}
    \caption{Illustration of the $(K-1)$-separable box in \eqref{hardy_fam}. Solid lines represent local deterministic links and the box $P$ is constructed from the convex mixture of boxes where the party $A_1$ shares hardy correlations $P_H$ with only one other party and local deterministic correlations with the rest. }
    \label{fig4}
\end{figure}

\begin{prop}\label{prop1}
 In a Bell setup with $K$ parties $\{A_k\}_{k=1}^{K}$
 there exists a family of $(K-1)$-separable quantum distributions
\begin{align}\nonumber
    P(a_1,...,a_K|x_1,...,x_K) &= \frac{1}{K-1}\sum_{k=2}^{K}P_H(a_1^k,a_k|x_1^k,x_k)\\&\prod_{m\in\{2,\ldots,K\}\setminus\{k\}}\delta_{1,a_1^{m}}\delta_{1,a_m}\label{hardy_fam}
\end{align}
such that when $K-1$ copies are wired with the AND-gated wiring (see Figure-\ref{fig2} for the $K=3$ example) given by
\begin{align}\label{and_wiring}
&a^{k}_1 = a_{(1)1}^k\land a_{(2)1}^k\land ... \land a_{(K-1)1}^k\\ \nonumber&a_{k} = a_{(1)k}\land a_{(2)k} \land... \land a_{(K-1)k}\\ \nonumber &x_1^{k}=x_{(1)1}^k = x_{(2)1}^k = ... = x_{(K-1)1}^k\\ \nonumber & x_k = x_{(1)k} = x_{(2)k}=... = x_{(K-1)k}.
\end{align}
The resulting distribution $P^\prime(a_1,...,a_K|x_1,...,x_K)$ is $K$-party GMNL.
\end{prop}

Here we provide the proof for the simplest case with $K=3$; the general case is given in Appendix~\ref{app:prop1}.

\begin{proof}[Proof for K=3]
 Party $A_1$ has 4 inputs and outputs while the other two parties have 2 inputs and outputs each. The initial biseparable quantum distribution is given by
\begin{align}\label{ex1}
    \nonumber P(a_1,a_2,a_3|x_1,x_2,x_3) &= \frac{1}{2}P_1(a_1,a_2,a_3|x_1,x_2,x_3)\\
    &+\frac{1}{2}P_2(a_1,a_2,a_3|x_1,x_2,x_3)
\end{align}
where $P_1=P_H(a_1^2,a_2|x^2_1,x_2)\delta_{1,a_1^{3}}\delta_{1,a_3}$ and $P_2=P_H(a_1^3,a_3|x^3_1,x_3)\delta_{1,a_1^{2}}\delta_{1,a_2}$. In this case, the witness~\eqref{top_ineq} is given by
\begin{align}\label{ex_ineq}
    I_3(P) = &I(Q^{P}_2) + I(Q^P_3)+P(00,0,0|00,0,0)\\&-Q^P_2(0,0|0,0)-Q^P_3(0,0|0,0)\leq 0. \nonumber
\end{align}

\begin{figure}[t]
    \centering
    \includegraphics[width=1.0\columnwidth]{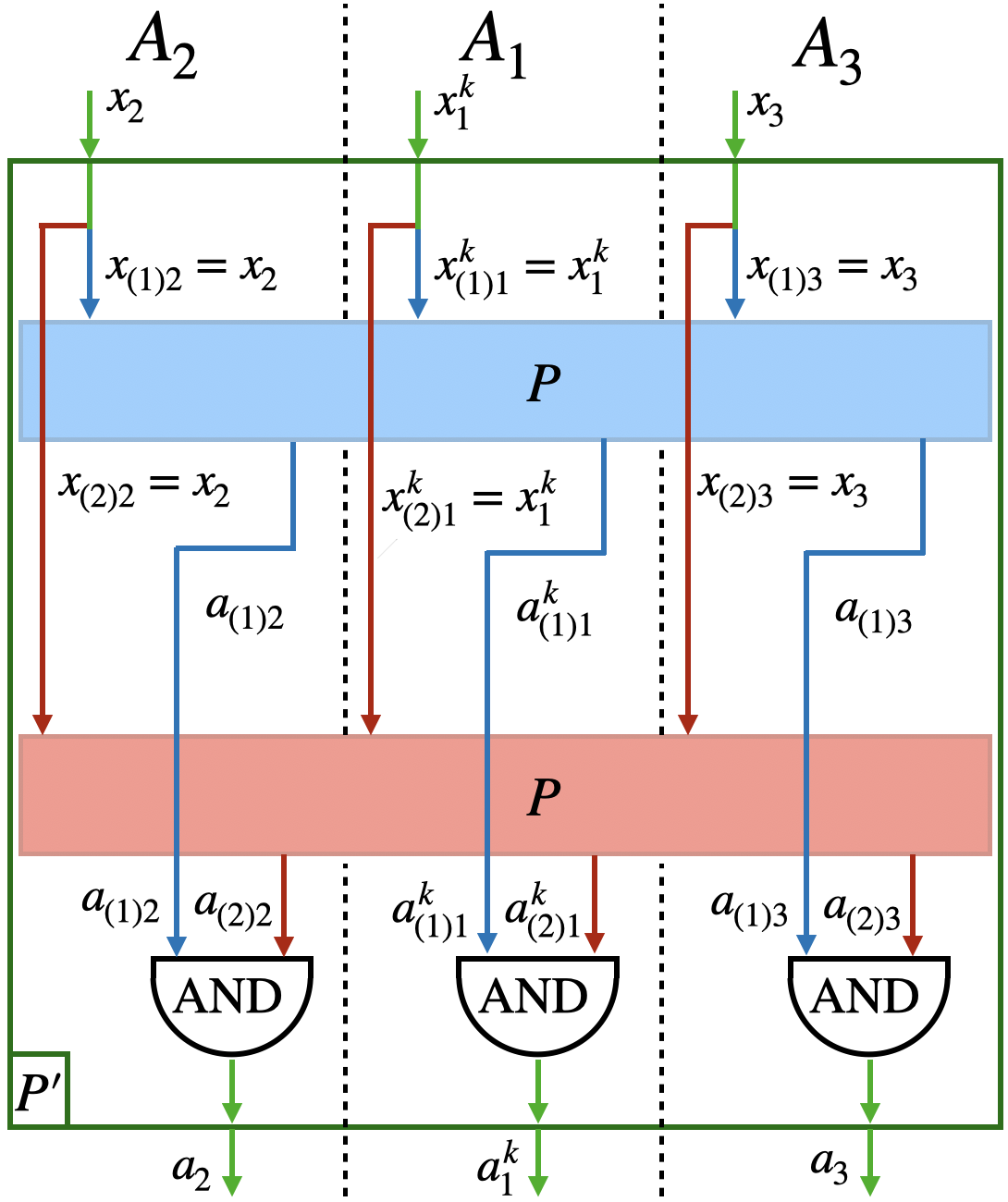}
    \caption{Operational depiction of the AND wiring: The blue (red) box represents the first (second) copy of the distribution $P$. The blue (red) arrows depict the inputs and outputs of the first (second) copy of $P$. The green box represents the distribution $P'$ obtained as a result of the wiring while the green arrows represent the inputs and outputs of this box. The dashed lines represent the space-like separation of the parties $A_1,A_2$ and $A_3$.}
    \label{fig2}
\end{figure}

We denote the action of the AND wiring on two copies using the bilinear map $\mathcal{F}$ such that $\mathcal{F}[P,P]= P^\prime$. By bilinearity of $\mathcal{F}$, we have $P^\prime = \frac{1}{4}\mathcal{F}[P_1,P_1] + \frac{1}{2}\mathcal{F}[P_1,P_2] + \frac{1}{4}\mathcal{F}[P_2,P_2]$, where
\begin{align}
    \mathcal{F}[P_1,P_1] &= \Tilde{P}(a_1^2,a_2|x^2_1,x_2)\delta_{1,a_1^{3}}\delta_{1,a_3}\\
    \mathcal{F}[P_2,P_2] &= \Tilde{P}(a_1^3,a_3|x^3_1,x_3)\delta_{1,a_1^{2}}\delta_{1,a_2}\\
    \mathcal{F}[P_1,P_2] &=P_H(a_1^2,a_2|x^2_1,x_2)P_H(a_1^3,a_3|x^3_1,x_3).
\end{align}
Note that this uses the symmetry of the wiring under swap of the two distributions, i.e., $\mathcal{F}[P_1,P_2]=\mathcal{F}[P_2,P_1]$.  It is easy to check that no matter what the distribution $\Tilde{P}$ is we have $I_3\left(\mathcal{F}[P_1,P_1]\right) = I_3\left(\mathcal{F}[P_2,P_2]\right) = 0$. Therefore, the term we are concerned with is $P^*=\mathcal{F}[P_1,P_2]$ meaning $I_3(P^\prime) = \frac{1}{2}I_3\left(P^*\right)$. First, notice that the negative terms in \eqref{ex_ineq} cancel out with the positive terms in $I\left(Q_2^{P^*}\right)$ and $I\left(Q_3^{P^*}\right)$. Furthermore, due to Hardy correlations having the property $P_H(01|01)=P_H(10|10)=P_H(00|11)=0$ and $P_H(00|00)=p(t)$ we see that the negative terms in  $I\left(Q_2^{P^*}\right)$ and $I\left(Q_3^{P^*}\right)$ are zero and the only nonzero element left in $I_3(P^*)$ is $P^* (00,0,0|00,0,0) = P_H(0,0|0,0)P_H(0,0|0,0)= p^2(t)$. Thus, $I_3(P^\prime)=\frac{1}{2}p^2(t)>0$ and the new distribution $P^\prime$ is GMNL.
\end{proof}

\section{Superactivation of GMNL via Catalysis}

In the super-activation phenomenon presented above, two copies of an original $(K-1)$-separable box are used to obtain a final box that is GMNL. In this process, both copies of the original box are used, hence in the end all boxes have been consumed. Here we show that this is not always necessary, by demonstrating an alternative activation effect based on catalysis. 

Specifically, we start from a single copy of a $(K-1)$-separable box $P$. This box will be processed with another box, a ``catalyst'' box $P_c$. By processing these two boxes locally, the parties can then obtain a final box that is GMNL, while at the same time returning the catalyst undisturbed. More precisely, the parties perform local operations to implement the transformation $P\otimes P_c\rightarrow P'_{\mathrm{tot}}$. Let us now define the marginal distributions of $P'$, i.e.
\begin{align}
&P'_{\mathrm{tar}}(\vec a'_{\mathrm{tar}}|\vec x'_{\mathrm{tar}}) = \sum_{\vec a'_{c}}P'_{\mathrm{tot}}(\vec a'_{\mathrm{tar}},\vec a^{\prime}_c|\vec x'_{\mathrm{tar}},\vec x^{\prime}_c),\\
&P'_c(\vec a^{\prime}_c|\vec x_c') = \sum_{\vec a'_{\mathrm{tar}}}P'_{\mathrm{tot}}(\vec a'_{\mathrm{tar}},\vec a^{\prime}_c|\vec x'_{\mathrm{tar}},\vec x^{\prime}_c),
\end{align}
where $\vec a'_{\mathrm{tar}} = (a'_1,...,a'_K)$, $\vec x'_{\mathrm{tar}} = (x'_1,...,x'_K)$ and $P'_c(\vec a^{\prime}_c|\vec x_c')$ where $\vec a'_{c} = (a^{c \prime}_1,...,a^{c \prime}_K)$, $\vec x'_{c} = (x^{c \prime}_1,...,x^{c \prime}_K)$. Finally, we obtain that (i) $P'_\mathrm{tar}$ is GMNL, and (ii) $P'_c = P_c$. Hence, the initial box $P$ has been transformed in the target distribution $P'_\mathrm{tar}$ which is GMNL, while the catalyst has been returned unchanged. Importantly, note that this protocol works at the price of introducing correlations between the target and the catalyst, i.e. $P'_{tot} \neq P'_{\mathrm{tar}}\otimes P'_c$ \footnote{This form of catalysis is usually referred as ``marginal catalysis''. One could in principle also consider strict catalysis, i.e. requiring $P'_{tot} = P'_{\mathrm{tar}}\otimes P'_c$, but we believe that activation would be impossible in this case, following the argument of Ref.~\cite{Karvonen_2021}}. An interesting feature is that the catalyst state is (also) biseparable. Hence, GMNL can be obtained by combining a $(K-1)$-separable box with a biseparable catalyst. More formally, we have the following result.

\begin{prop} \label{prop2}
   In a Bell scenario with $K$ parties $\{A_k\}_{k=1}^{K}$ let $P$ be a biseparable distribution in this scenario such that a minimum of $n$ copies of $P$ are needed for it to be turned into a GMNL distribution $P'$ and let the wiring $\mathcal{W}$ facilitate this transformation i.e., $\mathcal{W}[P,...,P]=P'$. Furthermore, let $I^*$ be a witness for this activation such that $I^*(P)\leq 0$, $I^*(P')>0$ and let $P_{*}$ be a local deterministic distribution such that $I^*(P_*)=0$. Then, for $P$, GMNL can be super-activated using the  biseparable catalyst 
\begin{align}\label{catalyst}
    P_c=\frac{1}{n}\sum_{i=0}^{n-1}\Big(P^{\otimes i}\times P_{*}^{\otimes(n-1-i)}\Big)\times\prod_{k=1}^K \delta_{i,\tilde a_k}.
\end{align}
 where $P^{\otimes i}$ denotes $i$ copies of the $K$-partite box $P$.
\end{prop}
 This means that from a resource point of view, a single quantum biseparable resource can be converted into a GMNL resource (with the help of a catalyst).

Notice that the non-local boxes in our scenarios are generated by measurements on quantum systems. Hence making an input into a box means choosing a measurement and consumes that subsystem. Thus, to be able to restore the catalyst to its original state, the classical processing by each party proceeds by looking up the classical register $\tilde{a}_k$ of the catalyst and then implementing a local classical pre- and post-processing operation of inputs and outputs for each box, leading to a new box (for which the inputs and outputs have not yet been chosen).

The proof for the general case is given in Appendix~\ref{app:catalyst}. Here we will present the proof for the minimal case. Intuitively, by looking at the $\tilde{a}$-registers of the catalyst each party ${A}_k$ can locally tell whether the $K$ parties jointly hold $n$-copies of $P$, in which case they can wire those copies to obtain $K$-party GMNL, or, whether they have fewer copies of $P$ in which case they choose their classical processing so as to restore the catalyst state.  

\begin{proof}[Proof for K=3] For brevity, let $\tilde a = (\tilde a_1, \tilde a_2, \tilde a_3)$ and $\delta_{i,\tilde a} = \delta_{i,\tilde a_1}\delta_{i,\tilde a_2}\delta_{i,\tilde a_3}$. The parties start with the original distribution $P(\vec a|\vec x)=P(a_1,a_2,a_3|x_1,x_2,x_3)$ and the catalyst distribution
\begin{align}
    P_c(\vec a_c,\tilde a|\vec x_c) = \frac{1}{2}P(\vec a_c|\vec x_c)\delta_{1,\tilde a} + \frac{1}{2}P_*(\vec a_c| \vec x_c)\delta_{0,\tilde a}
\end{align}
where $\vec a_c = (a_1^c,a_2^c,a_3^c)$ and $\vec x_c = (x_1^c,x_2^c,x_3^c)$. The goal is to obtain a distribution $P_{\mathrm{tot}}'(\vec a',\vec a_c',\tilde a'|\vec x',\vec x_c')$ such that each marginal $P_c'(\vec a_c',\tilde a'|\vec x_c')=\sum_{\vec a'}P'_{\mathrm{tot}}(\vec a',\vec a_c',\tilde a'|\vec x',\vec x_c')$ and $P'_{\mathrm{tar}}(\vec a')=\sum_{\vec a_c',\tilde a'}P'_{\mathrm{tot}}(\vec a',\vec a_c',\tilde a'|\vec x',\vec x_c')$ satisfy $P_c' = P_c$ and $I^*(P_{\mathrm{tar}}')>0$ using only local processing of inputs and outcomes.

Notice that the parties always observe either $\tilde a = (0,0,0)$ or $\tilde a = (1,1,1)$. 

When they observe $\tilde a = (1,1,1)$ they know that the rest of their catalyst inputs and outputs follow the distribution $P(\vec a_c | \vec x_c)$. Since, $P_*$ is a local deterministic distribution we can express $P_*(\vec a_c|\vec x_c) = \delta_{a_1^c,f_1(x_1^c)} \delta_{a_2^c,f_2(x_2^c)} \delta_{a_3^c,f_3(x_3^c)}$ for local deterministic $f_1,f_2,f_3$ (see Section~\ref{sec:preliminaries}). The parties then use the new catalyst inputs $\vec x_c'$ to locally implement the processing ${a_i^c}'=f_i({x_i^c}') \ \forall i$ and they set $\tilde a'= (0,0,0)$. They further use the wiring $\mathcal{W}$ on their initial distribution and catalyst as
\begin{align}
\mathcal{W}\left[P(\vec a|\vec x),P(\vec a_c|\vec x_c)\right] = P'(\vec a'|\vec x').
\end{align}
 This generates the distribution
\begin{align}
    P_{(1,1,1)}(\vec a' \vec a_c' \tilde a' | \vec x' \vec x_c')= P'(\vec a'|\vec x') P_*(\vec a_c' | \vec x_c ') \delta_{0,\tilde a'} . 
\end{align}

If the parties observe $\tilde a = (0,0,0)$, they use the new catalyst inputs $\vec x_c'$ as the inputs $\vec x $ to the original box $P$, i.e. $\vec x = \vec x_c'$, and then set the outputs $\vec a_c'$ of the new catalyst as the outputs $\vec a$ of the original box, i.e. $\vec a_c' = \vec a$ and they locally set $\tilde a' = (1,1,1)$. They further set the target outputs ${a_i}'=f_i({x_i}') \ \forall i$. This leads to the distribution
\begin{align}
    P_{(0,0,0)}(\vec a' \vec a_c' \tilde a' | \vec x' \vec x_c)= P_*(\vec a'|\vec x') P(\vec a_c'|\vec x_c') \delta_{1,\tilde a'} . 
\end{align}

 Since $P_c(\tilde a = (0,0,0))=P_c(\tilde a = (1,1,1)) = 1/2$, the overall distribution after local processing becomes 
\begin{align}
    P_{{\mathrm{tot}}}'(\vec a' \vec a_c' \tilde a' | \vec x' \vec x_c')= &\frac{1}{2} P'(\vec a'|\vec x') P_*(\vec a_c' | \vec x_c ') \delta_{0,\tilde a'}\\
    +&\frac{1}{2} P_*(\vec a'|\vec x') P(\vec a_c'|\vec x_c') \delta_{1,\tilde a'} . 
\end{align}
Straightforward marginalisation shows that this recovers the catalyst box, $P'_c = P_c$.
Furthermore, since by assumption $I^*(P') > 0$ and $I^*(P_*) = 0$, the target box satisfies $I^*(P'_{\mathrm{tar}})=I^*(P')/2>0$.
\end{proof}

A specific example of this with the GMNL inequality \eqref{top_ineq}, the distribution \eqref{ex1} and the AND wiring in Proposition-\ref{prop1} is provided in Appendix-\ref{app:new_wiring} for additional illustration.

\section{Discussion}

We have investigated the commonly used definition of GMNL proposed in Ref. \cite{Bancal2013} and shown that it can exhibit a phenomenon of superactivation. That is, by judiciously wiring two copies of a nonlocal box that is biseparable (i.e. not GMNL), it is possible to obtain a nonlocal box that is GMNL. The process is deterministic and involves only local wirings of the boxes. In particular, after presenting an illustrative example for the tripartite scenario, we have demonstrated this effect in a more general setting, even when restricting to quantum-realizable nonlocal boxes. For any number of parties $K$, we have constructed a quantum nonlocal box that is $(K-1)$-separable (i.e. almost fully separable), yet two copies wired together lead to a new quantum box that is GMNL. 
A key aspect in our examples of superactivation is that we construct nonlocal boxes that are convex mixtures of boxes that are separable across different partitions, such that the overall state is not separable across any bi-partition. An interesting open question is whether GMNL superactivation is possible for any nonlocal box of this form, when sufficiently many copies are wired together.

In the final part of the paper, we have also shown that GMNL can be activated at the single-copy level using a catalytic protocol based exclusively on local operations, and using a catalyst prepared in a state that is also biseparable. This effect is conceptually different from the catalytic activation of quantum Bell nonlocality recently reported in Ref. \cite{Bavaresco_2025}. In that paper, it is the nonlocality of specific entangled states (admitting a local model) that is activated. In contrast, our work considers activation at the level of nonlocal boxes (i.e. input/output multipartite distributions). An interesting question is whether this effect of catalysis (at the level of boxes) could enhance the process of nonlocality distillation.

From a more general perspective, our results question the operational relevance of the definition of GMNL of Ref. \cite{Bancal2013}, which is commonly used in the literature. Indeed, from a resource-theoretic perspective, it is desirable to have a notion of biseparability that is ``tensor stable'', i.e. that biseparable boxes remain biseparable even when many copies of the box are processed jointly. We have shown that this is not the case for the usual definition of GMNL of Ref. \cite{Bancal2013}. Note that, in this sense, our results are complementary to those of Refs \cite{Yamasaki2022,Palazuelos2022,weinbrenner2025}, which have shown that the usual definition of genuine multipartite entanglement (GME) suffers from a similar problem, namely it is not tensor-stable under Local Operations and Classical Communication (LOCC). Notice further that in contrast to the case of entanglement, allowing communication would not make sense in the case of nonlocality considered here. 

This motivates finding alternative definitions of GMNL (and GME). In fact, an interesting possibility was proposed in recent works \cite{Navascues2020,kraft2021,CoiteuxRoy2021,Coiteux2021pra} developing an approach based on ideas from the field of network Bell nonlocality. In this approach biseperable boxes (or quantum states) are defined based on an underlying network structure that involves only sources that connect strict subsets of the parties. Such notions of network-GMNL (and network-GME) are tensor-stable by construction, and therefore more satisfactory from an operational perspective. In particular, the $K$-party network that has a source for all subsets of $(K-1)$ parties provides a direct network alternative for the definition of GMNL. Whether network-GMNL remains stable under catalysis is an interesting open question.

\section{Acknowledgements}

We thank Victor Barizien and Eliot Donnadieu for discussions. The authors acknowledge financial support from the Swiss National Science Foundation, via project 236580.
MW is supported by the STeP2 grant (ANR-22-EXES-0013) of the Agence Nationale de la Recherche (ANR) as part of Plan France 2030.

\bibliographystyle{apsrev4-1}
\bibliography{main.bib}

\begin{thebibliography}{49}%
\makeatletter
\providecommand \@ifxundefined [1]{%
 \@ifx{#1\undefined}
}%
\providecommand \@ifnum [1]{%
 \ifnum #1\expandafter \@firstoftwo
 \else \expandafter \@secondoftwo
 \fi
}%
\providecommand \@ifx [1]{%
 \ifx #1\expandafter \@firstoftwo
 \else \expandafter \@secondoftwo
 \fi
}%
\providecommand \natexlab [1]{#1}%
\providecommand \enquote  [1]{``#1''}%
\providecommand \bibnamefont  [1]{#1}%
\providecommand \bibfnamefont [1]{#1}%
\providecommand \citenamefont [1]{#1}%
\providecommand \href@noop [0]{\@secondoftwo}%
\providecommand \href [0]{\begingroup \@sanitize@url \@href}%
\providecommand \@href[1]{\@@startlink{#1}\@@href}%
\providecommand \@@href[1]{\endgroup#1\@@endlink}%
\providecommand \@sanitize@url [0]{\catcode `\\12\catcode `\$12\catcode `\&12\catcode `\#12\catcode `\^12\catcode `\_12\catcode `\%12\relax}%
\providecommand \@@startlink[1]{}%
\providecommand \@@endlink[0]{}%
\providecommand \url  [0]{\begingroup\@sanitize@url \@url }%
\providecommand \@url [1]{\endgroup\@href {#1}{\urlprefix }}%
\providecommand \urlprefix  [0]{URL }%
\providecommand \Eprint [0]{\href }%
\providecommand \doibase [0]{http://dx.doi.org/}%
\providecommand \selectlanguage [0]{\@gobble}%
\providecommand \bibinfo  [0]{\@secondoftwo}%
\providecommand \bibfield  [0]{\@secondoftwo}%
\providecommand \translation [1]{[#1]}%
\providecommand \BibitemOpen [0]{}%
\providecommand \bibitemStop [0]{}%
\providecommand \bibitemNoStop [0]{.\EOS\space}%
\providecommand \EOS [0]{\spacefactor3000\relax}%
\providecommand \BibitemShut  [1]{\csname bibitem#1\endcsname}%
\let\auto@bib@innerbib\@empty
\bibitem [{\citenamefont {Scarani}\ and\ \citenamefont {Gisin}(2001)}]{Scarani2001}%
  \BibitemOpen
  \bibfield  {author} {\bibinfo {author} {\bibfnamefont {V.}~\bibnamefont {Scarani}}\ and\ \bibinfo {author} {\bibfnamefont {N.}~\bibnamefont {Gisin}},\ }\href {\doibase 10.1103/PhysRevLett.87.117901} {\bibfield  {journal} {\bibinfo  {journal} {Phys. Rev. Lett.}\ }\textbf {\bibinfo {volume} {87}},\ \bibinfo {pages} {117901} (\bibinfo {year} {2001})}\BibitemShut {NoStop}%
\bibitem [{\citenamefont {Ribeiro}\ \emph {et~al.}(2018)\citenamefont {Ribeiro}, \citenamefont {Murta},\ and\ \citenamefont {Wehner}}]{Ribeiro2018}%
  \BibitemOpen
  \bibfield  {author} {\bibinfo {author} {\bibfnamefont {J.}~\bibnamefont {Ribeiro}}, \bibinfo {author} {\bibfnamefont {G.}~\bibnamefont {Murta}}, \ and\ \bibinfo {author} {\bibfnamefont {S.}~\bibnamefont {Wehner}},\ }\href {\doibase 10.1103/PhysRevA.97.022307} {\bibfield  {journal} {\bibinfo  {journal} {Phys. Rev. A}\ }\textbf {\bibinfo {volume} {97}},\ \bibinfo {pages} {022307} (\bibinfo {year} {2018})}\BibitemShut {NoStop}%
\bibitem [{\citenamefont {Holz}\ \emph {et~al.}(2020)\citenamefont {Holz}, \citenamefont {Kampermann},\ and\ \citenamefont {Bru\ss{}}}]{Holz2020}%
  \BibitemOpen
  \bibfield  {author} {\bibinfo {author} {\bibfnamefont {T.}~\bibnamefont {Holz}}, \bibinfo {author} {\bibfnamefont {H.}~\bibnamefont {Kampermann}}, \ and\ \bibinfo {author} {\bibfnamefont {D.}~\bibnamefont {Bru\ss{}}},\ }\href {\doibase 10.1103/PhysRevResearch.2.023251} {\bibfield  {journal} {\bibinfo  {journal} {Phys. Rev. Res.}\ }\textbf {\bibinfo {volume} {2}},\ \bibinfo {pages} {023251} (\bibinfo {year} {2020})}\BibitemShut {NoStop}%
\bibitem [{\citenamefont {Grasselli}\ \emph {et~al.}(2023)\citenamefont {Grasselli}, \citenamefont {Murta}, \citenamefont {Kampermann},\ and\ \citenamefont {Bruß}}]{Grasselli_2023}%
  \BibitemOpen
  \bibfield  {author} {\bibinfo {author} {\bibfnamefont {F.}~\bibnamefont {Grasselli}}, \bibinfo {author} {\bibfnamefont {G.}~\bibnamefont {Murta}}, \bibinfo {author} {\bibfnamefont {H.}~\bibnamefont {Kampermann}}, \ and\ \bibinfo {author} {\bibfnamefont {D.}~\bibnamefont {Bruß}},\ }\href {\doibase 10.22331/q-2023-04-13-980} {\bibfield  {journal} {\bibinfo  {journal} {Quantum}\ }\textbf {\bibinfo {volume} {7}},\ \bibinfo {pages} {980} (\bibinfo {year} {2023})}\BibitemShut {NoStop}%
\bibitem [{\citenamefont {Bancal}\ \emph {et~al.}(2011{\natexlab{a}})\citenamefont {Bancal}, \citenamefont {Gisin}, \citenamefont {Liang},\ and\ \citenamefont {Pironio}}]{bancal2011_diew}%
  \BibitemOpen
  \bibfield  {author} {\bibinfo {author} {\bibfnamefont {J.-D.}\ \bibnamefont {Bancal}}, \bibinfo {author} {\bibfnamefont {N.}~\bibnamefont {Gisin}}, \bibinfo {author} {\bibfnamefont {Y.-C.}\ \bibnamefont {Liang}}, \ and\ \bibinfo {author} {\bibfnamefont {S.}~\bibnamefont {Pironio}},\ }\href {\doibase 10.1103/PhysRevLett.106.250404} {\bibfield  {journal} {\bibinfo  {journal} {Phys. Rev. Lett.}\ }\textbf {\bibinfo {volume} {106}},\ \bibinfo {pages} {250404} (\bibinfo {year} {2011}{\natexlab{a}})}\BibitemShut {NoStop}%
\bibitem [{\citenamefont {Svetlichny}(1987)}]{svetlichny87}%
  \BibitemOpen
  \bibfield  {author} {\bibinfo {author} {\bibfnamefont {G.}~\bibnamefont {Svetlichny}},\ }\href {\doibase 10.1103/PhysRevD.35.3066} {\bibfield  {journal} {\bibinfo  {journal} {Phys. Rev. D}\ }\textbf {\bibinfo {volume} {35}},\ \bibinfo {pages} {3066} (\bibinfo {year} {1987})}\BibitemShut {NoStop}%
\bibitem [{\citenamefont {Seevinck}\ and\ \citenamefont {Svetlichny}(2002)}]{Seevinck2002}%
  \BibitemOpen
  \bibfield  {author} {\bibinfo {author} {\bibfnamefont {M.}~\bibnamefont {Seevinck}}\ and\ \bibinfo {author} {\bibfnamefont {G.}~\bibnamefont {Svetlichny}},\ }\href {\doibase 10.1103/PhysRevLett.89.060401} {\bibfield  {journal} {\bibinfo  {journal} {Phys. Rev. Lett.}\ }\textbf {\bibinfo {volume} {89}},\ \bibinfo {pages} {060401} (\bibinfo {year} {2002})}\BibitemShut {NoStop}%
\bibitem [{\citenamefont {Collins}\ \emph {et~al.}(2002)\citenamefont {Collins}, \citenamefont {Gisin}, \citenamefont {Popescu}, \citenamefont {Roberts},\ and\ \citenamefont {Scarani}}]{Collins2002}%
  \BibitemOpen
  \bibfield  {author} {\bibinfo {author} {\bibfnamefont {D.}~\bibnamefont {Collins}}, \bibinfo {author} {\bibfnamefont {N.}~\bibnamefont {Gisin}}, \bibinfo {author} {\bibfnamefont {S.}~\bibnamefont {Popescu}}, \bibinfo {author} {\bibfnamefont {D.}~\bibnamefont {Roberts}}, \ and\ \bibinfo {author} {\bibfnamefont {V.}~\bibnamefont {Scarani}},\ }\href {\doibase 10.1103/PhysRevLett.88.170405} {\bibfield  {journal} {\bibinfo  {journal} {Phys. Rev. Lett.}\ }\textbf {\bibinfo {volume} {88}},\ \bibinfo {pages} {170405} (\bibinfo {year} {2002})}\BibitemShut {NoStop}%
\bibitem [{\citenamefont {Bancal}\ \emph {et~al.}(2011{\natexlab{b}})\citenamefont {Bancal}, \citenamefont {Brunner}, \citenamefont {Gisin},\ and\ \citenamefont {Liang}}]{Bancal2011}%
  \BibitemOpen
  \bibfield  {author} {\bibinfo {author} {\bibfnamefont {J.-D.}\ \bibnamefont {Bancal}}, \bibinfo {author} {\bibfnamefont {N.}~\bibnamefont {Brunner}}, \bibinfo {author} {\bibfnamefont {N.}~\bibnamefont {Gisin}}, \ and\ \bibinfo {author} {\bibfnamefont {Y.-C.}\ \bibnamefont {Liang}},\ }\href {\doibase 10.1103/PhysRevLett.106.020405} {\bibfield  {journal} {\bibinfo  {journal} {Phys. Rev. Lett.}\ }\textbf {\bibinfo {volume} {106}},\ \bibinfo {pages} {020405} (\bibinfo {year} {2011}{\natexlab{b}})}\BibitemShut {NoStop}%
\bibitem [{\citenamefont {Lavoie}\ \emph {et~al.}(2009)\citenamefont {Lavoie}, \citenamefont {Kaltenbaek},\ and\ \citenamefont {Resch}}]{Lavoie2009}%
  \BibitemOpen
  \bibfield  {author} {\bibinfo {author} {\bibfnamefont {J.}~\bibnamefont {Lavoie}}, \bibinfo {author} {\bibfnamefont {R.}~\bibnamefont {Kaltenbaek}}, \ and\ \bibinfo {author} {\bibfnamefont {K.~J.}\ \bibnamefont {Resch}},\ }\href {\doibase 10.1088/1367-2630/11/7/073051} {\bibfield  {journal} {\bibinfo  {journal} {New Journal of Physics}\ }\textbf {\bibinfo {volume} {11}},\ \bibinfo {pages} {073051} (\bibinfo {year} {2009})}\BibitemShut {NoStop}%
\bibitem [{\citenamefont {Gallego}\ \emph {et~al.}(2012)\citenamefont {Gallego}, \citenamefont {W\"urflinger}, \citenamefont {Ac\'{\i}n},\ and\ \citenamefont {Navascu\'es}}]{Gallego2012}%
  \BibitemOpen
  \bibfield  {author} {\bibinfo {author} {\bibfnamefont {R.}~\bibnamefont {Gallego}}, \bibinfo {author} {\bibfnamefont {L.~E.}\ \bibnamefont {W\"urflinger}}, \bibinfo {author} {\bibfnamefont {A.}~\bibnamefont {Ac\'{\i}n}}, \ and\ \bibinfo {author} {\bibfnamefont {M.}~\bibnamefont {Navascu\'es}},\ }\href {\doibase 10.1103/PhysRevLett.109.070401} {\bibfield  {journal} {\bibinfo  {journal} {Phys. Rev. Lett.}\ }\textbf {\bibinfo {volume} {109}},\ \bibinfo {pages} {070401} (\bibinfo {year} {2012})}\BibitemShut {NoStop}%
\bibitem [{\citenamefont {Bancal}\ \emph {et~al.}(2013)\citenamefont {Bancal}, \citenamefont {Barrett}, \citenamefont {Gisin},\ and\ \citenamefont {Pironio}}]{Bancal2013}%
  \BibitemOpen
  \bibfield  {author} {\bibinfo {author} {\bibfnamefont {J.-D.}\ \bibnamefont {Bancal}}, \bibinfo {author} {\bibfnamefont {J.}~\bibnamefont {Barrett}}, \bibinfo {author} {\bibfnamefont {N.}~\bibnamefont {Gisin}}, \ and\ \bibinfo {author} {\bibfnamefont {S.}~\bibnamefont {Pironio}},\ }\href {\doibase 10.1103/PhysRevA.88.014102} {\bibfield  {journal} {\bibinfo  {journal} {Phys. Rev. A}\ }\textbf {\bibinfo {volume} {88}},\ \bibinfo {pages} {014102} (\bibinfo {year} {2013})}\BibitemShut {NoStop}%
\bibitem [{\citenamefont {Baccari}\ \emph {et~al.}(2019)\citenamefont {Baccari}, \citenamefont {Tura}, \citenamefont {Fadel}, \citenamefont {Aloy}, \citenamefont {Bancal}, \citenamefont {Sangouard}, \citenamefont {Lewenstein}, \citenamefont {Ac\'{\i}n},\ and\ \citenamefont {Augusiak}}]{Baccari2019}%
  \BibitemOpen
  \bibfield  {author} {\bibinfo {author} {\bibfnamefont {F.}~\bibnamefont {Baccari}}, \bibinfo {author} {\bibfnamefont {J.}~\bibnamefont {Tura}}, \bibinfo {author} {\bibfnamefont {M.}~\bibnamefont {Fadel}}, \bibinfo {author} {\bibfnamefont {A.}~\bibnamefont {Aloy}}, \bibinfo {author} {\bibfnamefont {J.-D.}\ \bibnamefont {Bancal}}, \bibinfo {author} {\bibfnamefont {N.}~\bibnamefont {Sangouard}}, \bibinfo {author} {\bibfnamefont {M.}~\bibnamefont {Lewenstein}}, \bibinfo {author} {\bibfnamefont {A.}~\bibnamefont {Ac\'{\i}n}}, \ and\ \bibinfo {author} {\bibfnamefont {R.}~\bibnamefont {Augusiak}},\ }\href {\doibase 10.1103/PhysRevA.100.022121} {\bibfield  {journal} {\bibinfo  {journal} {Phys. Rev. A}\ }\textbf {\bibinfo {volume} {100}},\ \bibinfo {pages} {022121} (\bibinfo {year} {2019})}\BibitemShut {NoStop}%
\bibitem [{\citenamefont {Horodecki}\ and\ \citenamefont {Ramanathan}(2019)}]{Horodecki2019}%
  \BibitemOpen
  \bibfield  {author} {\bibinfo {author} {\bibfnamefont {P.}~\bibnamefont {Horodecki}}\ and\ \bibinfo {author} {\bibfnamefont {R.}~\bibnamefont {Ramanathan}},\ }\href@noop {} {\bibfield  {journal} {\bibinfo  {journal} {Nature Communications}\ }\textbf {\bibinfo {volume} {10}},\ \bibinfo {pages} {1701} (\bibinfo {year} {2019})}\BibitemShut {NoStop}%
\bibitem [{\citenamefont {Curchod}\ \emph {et~al.}(2019)\citenamefont {Curchod}, \citenamefont {Almeida},\ and\ \citenamefont {Acín}}]{Curchod_2019}%
  \BibitemOpen
  \bibfield  {author} {\bibinfo {author} {\bibfnamefont {F.~J.}\ \bibnamefont {Curchod}}, \bibinfo {author} {\bibfnamefont {M.~L.}\ \bibnamefont {Almeida}}, \ and\ \bibinfo {author} {\bibfnamefont {A.}~\bibnamefont {Acín}},\ }\href {\doibase 10.1088/1367-2630/aaff2d} {\bibfield  {journal} {\bibinfo  {journal} {New Journal of Physics}\ }\textbf {\bibinfo {volume} {21}},\ \bibinfo {pages} {023016} (\bibinfo {year} {2019})}\BibitemShut {NoStop}%
\bibitem [{\citenamefont {Contreras-Tejada}\ \emph {et~al.}(2021)\citenamefont {Contreras-Tejada}, \citenamefont {Palazuelos},\ and\ \citenamefont {de~Vicente}}]{Contreras_Tejada_2021}%
  \BibitemOpen
  \bibfield  {author} {\bibinfo {author} {\bibfnamefont {P.}~\bibnamefont {Contreras-Tejada}}, \bibinfo {author} {\bibfnamefont {C.}~\bibnamefont {Palazuelos}}, \ and\ \bibinfo {author} {\bibfnamefont {J.~I.}\ \bibnamefont {de~Vicente}},\ }\href {\doibase 10.1103/physrevlett.126.040501} {\bibfield  {journal} {\bibinfo  {journal} {Physical Review Letters}\ }\textbf {\bibinfo {volume} {126}} (\bibinfo {year} {2021}),\ 10.1103/physrevlett.126.040501}\BibitemShut {NoStop}%
\bibitem [{\citenamefont {Palazuelos}(2012)}]{palazuelos2012}%
  \BibitemOpen
  \bibfield  {author} {\bibinfo {author} {\bibfnamefont {C.}~\bibnamefont {Palazuelos}},\ }\href {\doibase 10.1103/PhysRevLett.109.190401} {\bibfield  {journal} {\bibinfo  {journal} {Phys. Rev. Lett.}\ }\textbf {\bibinfo {volume} {109}},\ \bibinfo {pages} {190401} (\bibinfo {year} {2012})}\BibitemShut {NoStop}%
\bibitem [{\citenamefont {Cavalcanti}\ \emph {et~al.}(2013)\citenamefont {Cavalcanti}, \citenamefont {Ac\'{\i}n}, \citenamefont {Brunner},\ and\ \citenamefont {V\'ertesi}}]{Cavalcanti2013}%
  \BibitemOpen
  \bibfield  {author} {\bibinfo {author} {\bibfnamefont {D.}~\bibnamefont {Cavalcanti}}, \bibinfo {author} {\bibfnamefont {A.}~\bibnamefont {Ac\'{\i}n}}, \bibinfo {author} {\bibfnamefont {N.}~\bibnamefont {Brunner}}, \ and\ \bibinfo {author} {\bibfnamefont {T.}~\bibnamefont {V\'ertesi}},\ }\href {\doibase 10.1103/PhysRevA.87.042104} {\bibfield  {journal} {\bibinfo  {journal} {Phys. Rev. A}\ }\textbf {\bibinfo {volume} {87}},\ \bibinfo {pages} {042104} (\bibinfo {year} {2013})}\BibitemShut {NoStop}%
\bibitem [{\citenamefont {Quintino}\ \emph {et~al.}(2016)\citenamefont {Quintino}, \citenamefont {Brunner},\ and\ \citenamefont {Huber}}]{Quintino2016}%
  \BibitemOpen
  \bibfield  {author} {\bibinfo {author} {\bibfnamefont {M.~T.}\ \bibnamefont {Quintino}}, \bibinfo {author} {\bibfnamefont {N.}~\bibnamefont {Brunner}}, \ and\ \bibinfo {author} {\bibfnamefont {M.}~\bibnamefont {Huber}},\ }\href {\doibase 10.1103/PhysRevA.94.062123} {\bibfield  {journal} {\bibinfo  {journal} {Phys. Rev. A}\ }\textbf {\bibinfo {volume} {94}},\ \bibinfo {pages} {062123} (\bibinfo {year} {2016})}\BibitemShut {NoStop}%
\bibitem [{\citenamefont {Miethlinger}\ \emph {et~al.}(2026)\citenamefont {Miethlinger}, \citenamefont {Castellano}, \citenamefont {Sekatski},\ and\ \citenamefont {Brunner}}]{miethlinger2026superactivationgenuinemultipartitebell}%
  \BibitemOpen
  \bibfield  {author} {\bibinfo {author} {\bibfnamefont {M.}~\bibnamefont {Miethlinger}}, \bibinfo {author} {\bibfnamefont {R.}~\bibnamefont {Castellano}}, \bibinfo {author} {\bibfnamefont {P.}~\bibnamefont {Sekatski}}, \ and\ \bibinfo {author} {\bibfnamefont {N.}~\bibnamefont {Brunner}},\ }\href {https://arxiv.org/abs/2603.17783} {\enquote {\bibinfo {title} {Superactivation of genuine multipartite bell nonlocality from two-party entanglement},}\ } (\bibinfo {year} {2026}),\ \Eprint {http://arxiv.org/abs/2603.17783} {arXiv:2603.17783 [quant-ph]} \BibitemShut {NoStop}%
\bibitem [{\citenamefont {Navascu\'es}\ \emph {et~al.}(2020)\citenamefont {Navascu\'es}, \citenamefont {Wolfe}, \citenamefont {Rosset},\ and\ \citenamefont {Pozas-Kerstjens}}]{Navascues2020}%
  \BibitemOpen
  \bibfield  {author} {\bibinfo {author} {\bibfnamefont {M.}~\bibnamefont {Navascu\'es}}, \bibinfo {author} {\bibfnamefont {E.}~\bibnamefont {Wolfe}}, \bibinfo {author} {\bibfnamefont {D.}~\bibnamefont {Rosset}}, \ and\ \bibinfo {author} {\bibfnamefont {A.}~\bibnamefont {Pozas-Kerstjens}},\ }\href {\doibase 10.1103/PhysRevLett.125.240505} {\bibfield  {journal} {\bibinfo  {journal} {Phys. Rev. Lett.}\ }\textbf {\bibinfo {volume} {125}},\ \bibinfo {pages} {240505} (\bibinfo {year} {2020})}\BibitemShut {NoStop}%
\bibitem [{\citenamefont {Kraft}\ \emph {et~al.}(2021)\citenamefont {Kraft}, \citenamefont {Designolle}, \citenamefont {Ritz}, \citenamefont {Brunner}, \citenamefont {G\"uhne},\ and\ \citenamefont {Huber}}]{kraft2021}%
  \BibitemOpen
  \bibfield  {author} {\bibinfo {author} {\bibfnamefont {T.}~\bibnamefont {Kraft}}, \bibinfo {author} {\bibfnamefont {S.}~\bibnamefont {Designolle}}, \bibinfo {author} {\bibfnamefont {C.}~\bibnamefont {Ritz}}, \bibinfo {author} {\bibfnamefont {N.}~\bibnamefont {Brunner}}, \bibinfo {author} {\bibfnamefont {O.}~\bibnamefont {G\"uhne}}, \ and\ \bibinfo {author} {\bibfnamefont {M.}~\bibnamefont {Huber}},\ }\href {\doibase 10.1103/PhysRevA.103.L060401} {\bibfield  {journal} {\bibinfo  {journal} {Phys. Rev. A}\ }\textbf {\bibinfo {volume} {103}},\ \bibinfo {pages} {L060401} (\bibinfo {year} {2021})}\BibitemShut {NoStop}%
\bibitem [{\citenamefont {Coiteux-Roy}\ \emph {et~al.}(2021{\natexlab{a}})\citenamefont {Coiteux-Roy}, \citenamefont {Wolfe},\ and\ \citenamefont {Renou}}]{CoiteuxRoy2021}%
  \BibitemOpen
  \bibfield  {author} {\bibinfo {author} {\bibfnamefont {X.}~\bibnamefont {Coiteux-Roy}}, \bibinfo {author} {\bibfnamefont {E.}~\bibnamefont {Wolfe}}, \ and\ \bibinfo {author} {\bibfnamefont {M.-O.}\ \bibnamefont {Renou}},\ }\href {\doibase 10.1103/PhysRevLett.127.200401} {\bibfield  {journal} {\bibinfo  {journal} {Phys. Rev. Lett.}\ }\textbf {\bibinfo {volume} {127}},\ \bibinfo {pages} {200401} (\bibinfo {year} {2021}{\natexlab{a}})}\BibitemShut {NoStop}%
\bibitem [{\citenamefont {Coiteux-Roy}\ \emph {et~al.}(2021{\natexlab{b}})\citenamefont {Coiteux-Roy}, \citenamefont {Wolfe},\ and\ \citenamefont {Renou}}]{Coiteux2021pra}%
  \BibitemOpen
  \bibfield  {author} {\bibinfo {author} {\bibfnamefont {X.}~\bibnamefont {Coiteux-Roy}}, \bibinfo {author} {\bibfnamefont {E.}~\bibnamefont {Wolfe}}, \ and\ \bibinfo {author} {\bibfnamefont {M.-O.}\ \bibnamefont {Renou}},\ }\href {\doibase 10.1103/PhysRevA.104.052207} {\bibfield  {journal} {\bibinfo  {journal} {Phys. Rev. A}\ }\textbf {\bibinfo {volume} {104}},\ \bibinfo {pages} {052207} (\bibinfo {year} {2021}{\natexlab{b}})}\BibitemShut {NoStop}%
\bibitem [{\citenamefont {Fritz}(2012)}]{Fritz2012}%
  \BibitemOpen
  \bibfield  {author} {\bibinfo {author} {\bibfnamefont {T.}~\bibnamefont {Fritz}},\ }\href {\doibase 10.1088/1367-2630/14/10/103001} {\bibfield  {journal} {\bibinfo  {journal} {New Journal of Physics}\ }\textbf {\bibinfo {volume} {14}},\ \bibinfo {pages} {103001} (\bibinfo {year} {2012})}\BibitemShut {NoStop}%
\bibitem [{\citenamefont {Tavakoli}\ \emph {et~al.}(2022)\citenamefont {Tavakoli}, \citenamefont {Pozas-Kerstjens}, \citenamefont {Luo},\ and\ \citenamefont {Renou}}]{Tavakoli2022}%
  \BibitemOpen
  \bibfield  {author} {\bibinfo {author} {\bibfnamefont {A.}~\bibnamefont {Tavakoli}}, \bibinfo {author} {\bibfnamefont {A.}~\bibnamefont {Pozas-Kerstjens}}, \bibinfo {author} {\bibfnamefont {M.-X.}\ \bibnamefont {Luo}}, \ and\ \bibinfo {author} {\bibfnamefont {M.-O.}\ \bibnamefont {Renou}},\ }\href {\doibase 10.1088/1361-6633/ac41bb} {\bibfield  {journal} {\bibinfo  {journal} {Reports on Progress in Physics}\ }\textbf {\bibinfo {volume} {85}},\ \bibinfo {pages} {056001} (\bibinfo {year} {2022})}\BibitemShut {NoStop}%
\bibitem [{Note1()}]{Note1}%
  \BibitemOpen
  \bibinfo {note} {Here, $ \protect \mathcal {X}_i$ and $\protect \mathcal {A}_i$ are considered to be finite sets.}\BibitemShut {Stop}%
\bibitem [{\citenamefont {Barrett}\ \emph {et~al.}(2005)\citenamefont {Barrett}, \citenamefont {Linden}, \citenamefont {Massar}, \citenamefont {Pironio}, \citenamefont {Popescu},\ and\ \citenamefont {Roberts}}]{Barrett_2005}%
  \BibitemOpen
  \bibfield  {author} {\bibinfo {author} {\bibfnamefont {J.}~\bibnamefont {Barrett}}, \bibinfo {author} {\bibfnamefont {N.}~\bibnamefont {Linden}}, \bibinfo {author} {\bibfnamefont {S.}~\bibnamefont {Massar}}, \bibinfo {author} {\bibfnamefont {S.}~\bibnamefont {Pironio}}, \bibinfo {author} {\bibfnamefont {S.}~\bibnamefont {Popescu}}, \ and\ \bibinfo {author} {\bibfnamefont {D.}~\bibnamefont {Roberts}},\ }\href {\doibase 10.1103/physreva.71.022101} {\bibfield  {journal} {\bibinfo  {journal} {Physical Review A}\ }\textbf {\bibinfo {volume} {71}} (\bibinfo {year} {2005}),\ 10.1103/physreva.71.022101}\BibitemShut {NoStop}%
\bibitem [{\citenamefont {Forster}\ \emph {et~al.}(2009)\citenamefont {Forster}, \citenamefont {Winkler},\ and\ \citenamefont {Wolf}}]{Forster_2009}%
  \BibitemOpen
  \bibfield  {author} {\bibinfo {author} {\bibfnamefont {M.}~\bibnamefont {Forster}}, \bibinfo {author} {\bibfnamefont {S.}~\bibnamefont {Winkler}}, \ and\ \bibinfo {author} {\bibfnamefont {S.}~\bibnamefont {Wolf}},\ }\href {\doibase 10.1103/physrevlett.102.120401} {\bibfield  {journal} {\bibinfo  {journal} {Physical Review Letters}\ }\textbf {\bibinfo {volume} {102}} (\bibinfo {year} {2009}),\ 10.1103/physrevlett.102.120401}\BibitemShut {NoStop}%
\bibitem [{\citenamefont {Brunner}\ and\ \citenamefont {Skrzypczyk}(2009)}]{Brunner_2009}%
  \BibitemOpen
  \bibfield  {author} {\bibinfo {author} {\bibfnamefont {N.}~\bibnamefont {Brunner}}\ and\ \bibinfo {author} {\bibfnamefont {P.}~\bibnamefont {Skrzypczyk}},\ }\href {\doibase 10.1103/physrevlett.102.160403} {\bibfield  {journal} {\bibinfo  {journal} {Physical Review Letters}\ }\textbf {\bibinfo {volume} {102}} (\bibinfo {year} {2009}),\ 10.1103/physrevlett.102.160403}\BibitemShut {NoStop}%
\bibitem [{\citenamefont {H\o{}yer}\ and\ \citenamefont {Rashid}(2010)}]{Hoyer_2010}%
  \BibitemOpen
  \bibfield  {author} {\bibinfo {author} {\bibfnamefont {P.}~\bibnamefont {H\o{}yer}}\ and\ \bibinfo {author} {\bibfnamefont {J.}~\bibnamefont {Rashid}},\ }\href {\doibase 10.1103/PhysRevA.82.042118} {\bibfield  {journal} {\bibinfo  {journal} {Phys. Rev. A}\ }\textbf {\bibinfo {volume} {82}},\ \bibinfo {pages} {042118} (\bibinfo {year} {2010})}\BibitemShut {NoStop}%
\bibitem [{\citenamefont {Brito}\ \emph {et~al.}(2019)\citenamefont {Brito}, \citenamefont {Moreno}, \citenamefont {Rai},\ and\ \citenamefont {Chaves}}]{Brito_2019}%
  \BibitemOpen
  \bibfield  {author} {\bibinfo {author} {\bibfnamefont {S.~G.~A.}\ \bibnamefont {Brito}}, \bibinfo {author} {\bibfnamefont {M.~G.~M.}\ \bibnamefont {Moreno}}, \bibinfo {author} {\bibfnamefont {A.}~\bibnamefont {Rai}}, \ and\ \bibinfo {author} {\bibfnamefont {R.}~\bibnamefont {Chaves}},\ }\href {\doibase 10.1103/physreva.100.012102} {\bibfield  {journal} {\bibinfo  {journal} {Physical Review A}\ }\textbf {\bibinfo {volume} {100}} (\bibinfo {year} {2019}),\ 10.1103/physreva.100.012102}\BibitemShut {NoStop}%
\bibitem [{\citenamefont {Eftaxias}\ \emph {et~al.}(2023)\citenamefont {Eftaxias}, \citenamefont {Weilenmann},\ and\ \citenamefont {Colbeck}}]{Eftaxias_2023}%
  \BibitemOpen
  \bibfield  {author} {\bibinfo {author} {\bibfnamefont {G.}~\bibnamefont {Eftaxias}}, \bibinfo {author} {\bibfnamefont {M.}~\bibnamefont {Weilenmann}}, \ and\ \bibinfo {author} {\bibfnamefont {R.}~\bibnamefont {Colbeck}},\ }\href {\doibase 10.1103/physrevlett.130.100201} {\bibfield  {journal} {\bibinfo  {journal} {Physical Review Letters}\ }\textbf {\bibinfo {volume} {130}} (\bibinfo {year} {2023}),\ 10.1103/physrevlett.130.100201}\BibitemShut {NoStop}%
\bibitem [{\citenamefont {Naik}\ \emph {et~al.}(2023)\citenamefont {Naik}, \citenamefont {Sidhardh}, \citenamefont {Sen}, \citenamefont {Roy}, \citenamefont {Rai},\ and\ \citenamefont {Banik}}]{Naik_2023}%
  \BibitemOpen
  \bibfield  {author} {\bibinfo {author} {\bibfnamefont {S.~G.}\ \bibnamefont {Naik}}, \bibinfo {author} {\bibfnamefont {G.~L.}\ \bibnamefont {Sidhardh}}, \bibinfo {author} {\bibfnamefont {S.}~\bibnamefont {Sen}}, \bibinfo {author} {\bibfnamefont {A.}~\bibnamefont {Roy}}, \bibinfo {author} {\bibfnamefont {A.}~\bibnamefont {Rai}}, \ and\ \bibinfo {author} {\bibfnamefont {M.}~\bibnamefont {Banik}},\ }\href {\doibase 10.1103/physrevlett.130.220201} {\bibfield  {journal} {\bibinfo  {journal} {Physical Review Letters}\ }\textbf {\bibinfo {volume} {130}} (\bibinfo {year} {2023}),\ 10.1103/physrevlett.130.220201}\BibitemShut {NoStop}%
\bibitem [{\citenamefont {Høyer}\ \emph {et~al.}(2026)\citenamefont {Høyer}, \citenamefont {Rashid},\ and\ \citenamefont {Ud~Din}}]{Hoyer_2026}%
  \BibitemOpen
  \bibfield  {author} {\bibinfo {author} {\bibfnamefont {P.}~\bibnamefont {Høyer}}, \bibinfo {author} {\bibfnamefont {J.}~\bibnamefont {Rashid}}, \ and\ \bibinfo {author} {\bibfnamefont {R.}~\bibnamefont {Ud~Din}},\ }in\ \href {\doibase 10.1109/qcnc69040.2026.00049} {\emph {\bibinfo {booktitle} {2026 International Conference on Quantum Communications, Networking, and Computing (QCNC)}}}\ (\bibinfo  {publisher} {IEEE},\ \bibinfo {year} {2026})\ p.\ \bibinfo {pages} {291–297}\BibitemShut {NoStop}%
\bibitem [{\citenamefont {Ulu}\ \emph {et~al.}(2025)\citenamefont {Ulu}, \citenamefont {Brunner},\ and\ \citenamefont {Weilenmann}}]{Ulu_2025}%
  \BibitemOpen
  \bibfield  {author} {\bibinfo {author} {\bibfnamefont {B.}~\bibnamefont {Ulu}}, \bibinfo {author} {\bibfnamefont {N.}~\bibnamefont {Brunner}}, \ and\ \bibinfo {author} {\bibfnamefont {M.}~\bibnamefont {Weilenmann}},\ }\href {\doibase 10.1103/f8jc-q1kg} {\bibfield  {journal} {\bibinfo  {journal} {Physical Review Letters}\ }\textbf {\bibinfo {volume} {135}} (\bibinfo {year} {2025}),\ 10.1103/f8jc-q1kg}\BibitemShut {NoStop}%
\bibitem [{\citenamefont {Tsirelson}(1993)}]{tsirelson}%
  \BibitemOpen
  \bibfield  {author} {\bibinfo {author} {\bibfnamefont {B.~S.}\ \bibnamefont {Tsirelson}},\ }\href@noop {} {\bibfield  {journal} {\bibinfo  {journal} {Hadronic Journal Supplement}\ }\textbf {\bibinfo {volume} {8}},\ \bibinfo {pages} {329} (\bibinfo {year} {1993})}\BibitemShut {NoStop}%
\bibitem [{\citenamefont {Popescu}\ and\ \citenamefont {Rohrlich}(1994)}]{PRbox}%
  \BibitemOpen
  \bibfield  {author} {\bibinfo {author} {\bibfnamefont {S.}~\bibnamefont {Popescu}}\ and\ \bibinfo {author} {\bibfnamefont {D.}~\bibnamefont {Rohrlich}},\ }\href {\doibase 10.1007/BF02058098} {\bibfield  {journal} {\bibinfo  {journal} {Foundations of Physics}\ }\textbf {\bibinfo {volume} {24}},\ \bibinfo {pages} {379} (\bibinfo {year} {1994})}\BibitemShut {NoStop}%
\bibitem [{\citenamefont {Pironio}\ \emph {et~al.}(2011)\citenamefont {Pironio}, \citenamefont {Bancal},\ and\ \citenamefont {Scarani}}]{Pironio_2011}%
  \BibitemOpen
  \bibfield  {author} {\bibinfo {author} {\bibfnamefont {S.}~\bibnamefont {Pironio}}, \bibinfo {author} {\bibfnamefont {J.-D.}\ \bibnamefont {Bancal}}, \ and\ \bibinfo {author} {\bibfnamefont {V.}~\bibnamefont {Scarani}},\ }\href {\doibase 10.1088/1751-8113/44/6/065303} {\bibfield  {journal} {\bibinfo  {journal} {Journal of Physics A: Mathematical and Theoretical}\ }\textbf {\bibinfo {volume} {44}},\ \bibinfo {pages} {065303} (\bibinfo {year} {2011})}\BibitemShut {NoStop}%
\bibitem [{\citenamefont {Hardy}(1993)}]{PhysRevLett.71.1665}%
  \BibitemOpen
  \bibfield  {author} {\bibinfo {author} {\bibfnamefont {L.}~\bibnamefont {Hardy}},\ }\href {\doibase 10.1103/PhysRevLett.71.1665} {\bibfield  {journal} {\bibinfo  {journal} {Phys. Rev. Lett.}\ }\textbf {\bibinfo {volume} {71}},\ \bibinfo {pages} {1665} (\bibinfo {year} {1993})}\BibitemShut {NoStop}%
\bibitem [{Note2()}]{Note2}%
  \BibitemOpen
  \bibinfo {note} {Note that the CH inequality is equivalent to the CHSH inequality as $I(P)=(I_{CHSH}(P)-2)/4$ \cite {Brunner_2014}.}\BibitemShut {Stop}%
\bibitem [{Note3()}]{Note3}%
  \BibitemOpen
  \bibinfo {note} {This form of catalysis is usually referred as ``marginal catalysis''. One could in principle also consider strict catalysis, i.e. requiring $P'_{tot} = P'_{\protect \mathrm {tar}}\otimes P'_c$, but we believe that activation would be impossible in this case, following the argument of Ref.~\cite {Karvonen_2021}}\BibitemShut {NoStop}%
\bibitem [{\citenamefont {Bavaresco}\ \emph {et~al.}(2025)\citenamefont {Bavaresco}, \citenamefont {Brunner}, \citenamefont {Girardin}, \citenamefont {Lipka-Bartosik},\ and\ \citenamefont {Sekatski}}]{Bavaresco_2025}%
  \BibitemOpen
  \bibfield  {author} {\bibinfo {author} {\bibfnamefont {J.}~\bibnamefont {Bavaresco}}, \bibinfo {author} {\bibfnamefont {N.}~\bibnamefont {Brunner}}, \bibinfo {author} {\bibfnamefont {A.}~\bibnamefont {Girardin}}, \bibinfo {author} {\bibfnamefont {P.}~\bibnamefont {Lipka-Bartosik}}, \ and\ \bibinfo {author} {\bibfnamefont {P.}~\bibnamefont {Sekatski}},\ }\href {\doibase 10.1103/5dth-7zm8} {\bibfield  {journal} {\bibinfo  {journal} {Physical Review Letters}\ }\textbf {\bibinfo {volume} {135}} (\bibinfo {year} {2025}),\ 10.1103/5dth-7zm8}\BibitemShut {NoStop}%
\bibitem [{\citenamefont {Yamasaki}\ \emph {et~al.}(2022)\citenamefont {Yamasaki}, \citenamefont {Morelli}, \citenamefont {Miethlinger}, \citenamefont {Bavaresco}, \citenamefont {Friis},\ and\ \citenamefont {Huber}}]{Yamasaki2022}%
  \BibitemOpen
  \bibfield  {author} {\bibinfo {author} {\bibfnamefont {H.}~\bibnamefont {Yamasaki}}, \bibinfo {author} {\bibfnamefont {S.}~\bibnamefont {Morelli}}, \bibinfo {author} {\bibfnamefont {M.}~\bibnamefont {Miethlinger}}, \bibinfo {author} {\bibfnamefont {J.}~\bibnamefont {Bavaresco}}, \bibinfo {author} {\bibfnamefont {N.}~\bibnamefont {Friis}}, \ and\ \bibinfo {author} {\bibfnamefont {M.}~\bibnamefont {Huber}},\ }\href {\doibase 10.22331/q-2022-04-25-695} {\bibfield  {journal} {\bibinfo  {journal} {{Quantum}}\ }\textbf {\bibinfo {volume} {6}},\ \bibinfo {pages} {695} (\bibinfo {year} {2022})}\BibitemShut {NoStop}%
\bibitem [{\citenamefont {Palazuelos}\ and\ \citenamefont {Vicente}(2022)}]{Palazuelos2022}%
  \BibitemOpen
  \bibfield  {author} {\bibinfo {author} {\bibfnamefont {C.}~\bibnamefont {Palazuelos}}\ and\ \bibinfo {author} {\bibfnamefont {J.~I.~d.}\ \bibnamefont {Vicente}},\ }\href {\doibase 10.22331/q-2022-06-13-735} {\bibfield  {journal} {\bibinfo  {journal} {{Quantum}}\ }\textbf {\bibinfo {volume} {6}},\ \bibinfo {pages} {735} (\bibinfo {year} {2022})}\BibitemShut {NoStop}%
\bibitem [{\citenamefont {Weinbrenner}\ \emph {et~al.}(2025)\citenamefont {Weinbrenner}, \citenamefont {Baksová}, \citenamefont {Denker}, \citenamefont {Morelli}, \citenamefont {Yu}, \citenamefont {Friis},\ and\ \citenamefont {Gühne}}]{weinbrenner2025}%
  \BibitemOpen
  \bibfield  {author} {\bibinfo {author} {\bibfnamefont {L.~T.}\ \bibnamefont {Weinbrenner}}, \bibinfo {author} {\bibfnamefont {K.}~\bibnamefont {Baksová}}, \bibinfo {author} {\bibfnamefont {S.}~\bibnamefont {Denker}}, \bibinfo {author} {\bibfnamefont {S.}~\bibnamefont {Morelli}}, \bibinfo {author} {\bibfnamefont {X.-D.}\ \bibnamefont {Yu}}, \bibinfo {author} {\bibfnamefont {N.}~\bibnamefont {Friis}}, \ and\ \bibinfo {author} {\bibfnamefont {O.}~\bibnamefont {Gühne}},\ }\href {https://arxiv.org/abs/2412.18331} {\enquote {\bibinfo {title} {Superactivation and incompressibility of genuine multipartite entanglement},}\ } (\bibinfo {year} {2025}),\ \Eprint {http://arxiv.org/abs/2412.18331} {arXiv:2412.18331 [quant-ph]} \BibitemShut {NoStop}%
\bibitem [{\citenamefont {Brunner}\ \emph {et~al.}(2014)\citenamefont {Brunner}, \citenamefont {Cavalcanti}, \citenamefont {Pironio}, \citenamefont {Scarani},\ and\ \citenamefont {Wehner}}]{Brunner_2014}%
  \BibitemOpen
  \bibfield  {author} {\bibinfo {author} {\bibfnamefont {N.}~\bibnamefont {Brunner}}, \bibinfo {author} {\bibfnamefont {D.}~\bibnamefont {Cavalcanti}}, \bibinfo {author} {\bibfnamefont {S.}~\bibnamefont {Pironio}}, \bibinfo {author} {\bibfnamefont {V.}~\bibnamefont {Scarani}}, \ and\ \bibinfo {author} {\bibfnamefont {S.}~\bibnamefont {Wehner}},\ }\href {\doibase 10.1103/revmodphys.86.419} {\bibfield  {journal} {\bibinfo  {journal} {Reviews of Modern Physics}\ }\textbf {\bibinfo {volume} {86}},\ \bibinfo {pages} {419–478} (\bibinfo {year} {2014})}\BibitemShut {NoStop}%
\bibitem [{\citenamefont {Karvonen}(2021)}]{Karvonen_2021}%
  \BibitemOpen
  \bibfield  {author} {\bibinfo {author} {\bibfnamefont {M.}~\bibnamefont {Karvonen}},\ }\href {\doibase 10.1103/physrevlett.127.160402} {\bibfield  {journal} {\bibinfo  {journal} {Physical Review Letters}\ }\textbf {\bibinfo {volume} {127}} (\bibinfo {year} {2021}),\ 10.1103/physrevlett.127.160402}\BibitemShut {NoStop}%
\bibitem [{\citenamefont {Pironio}(2005)}]{pironio_lifting}%
  \BibitemOpen
  \bibfield  {author} {\bibinfo {author} {\bibfnamefont {S.}~\bibnamefont {Pironio}},\ }\href {\doibase 10.1063/1.1928727} {\bibfield  {journal} {\bibinfo  {journal} {Journal of Mathematical Physics}\ }\textbf {\bibinfo {volume} {46}},\ \bibinfo {pages} {062112} (\bibinfo {year} {2005})}\BibitemShut {NoStop}%
\end{thebibliography}%

\appendix

\section{The tripartite GMNL witness $I^{169}_{\text{GMNL}}$}\label{app:explicit_ineq}
We define the correlators in this scenario as
\begin{align}\nonumber
    &\langle A_1^{x_1}A_2^{x_2}A_3^{x_3}\rangle \\&= \sum_{a_1,a_2,a_3}(-1)^{a_1+a_2+a_3}P(a_1,a_2,a_3|x_1,x_2,x_3)
\end{align}
such that the single-party and two-party correlators are similarly given by $\langle A_i^{x_i}\rangle = \sum_{a_i}(-1)^{a_i}P(a_i|x_i)$ and $\langle A_i^{x_i}A_j^{x_j}\rangle = \sum_{a_i,a_j}(-1)^{a_i + a_j}P(a_i,a_j|x_i,x_j)$ respectively where $i,j\in\{1,2,3\}$. Using these correlators we express the GMNL inequality of class 169 from \cite{Bancal2013} (after appropriate relabelling) as
\begin{align}\nonumber
  I^{169}_{\text{GMNL}}(P)=& -16 -   \langle A_1^{0}\rangle + \langle A_1^{1}\rangle -\langle A_2^{0}\rangle + \langle A_2^{1}\rangle - \langle A_1^{0}A_2^{1}\rangle\\\nonumber& - 3\langle A_1^{1}A_2^{0}\rangle +4\langle A_3^{0}\rangle + \langle A_1^{0}A_3^{0}\rangle + \langle A_1^{1}A_3^{0}\rangle\\\nonumber&-2\langle A_2^{0}A_3^{0}\rangle + 3\langle A_1^{0}A_2^{0}A_3^{0}\rangle + 5\langle A_1^{1}A_2^{0}A_3^{0}\rangle\\\nonumber&
  +4\langle A_1^{0}A_2^{1}A_3^{0}\rangle -4\langle A_1^{1}A_2^{1}A_3^{0}\rangle-2\langle A_1^{1}A_3^{1}\rangle\\\nonumber& +\langle A_2^{0}A_3^{1}\rangle+\langle A_2^{1}A_3^{1}\rangle + 3\langle A_1^{0}A_2^{0}A_3^{1}\rangle\\\nonumber & -2\langle A_1^{1}A_2^{0}A_3^{1}\rangle+3\langle A_1^{0}A_2^{1}A_3^{1}\rangle+4\langle A_1^{1}A_2^{1}A_3^{1}\rangle\\&\leq0.
\end{align}

\section{Constructing the GMNL witnesses}
\label{app:lifting}
This appendix follows the procedure introduced in~\cite{Curchod_2019} to construct the GNML witness \eqref{top_ineq}. This witness was already known~\cite{Contreras_Tejada_2021} and is rederived here for completeness.

This procedure is based on lifting Bell inequalities~\cite{pironio_lifting}, which is a way of turning a Bell inequality for a small number of parties inputs and outcomes into an inequality for a larger scenario with more parties, inputs or outcomes by embedding the original Bell test into such a larger scenario while fixing the extra parties to specific inputs and outputs. Nonlocality in the original test  implies nonlocality in the larger scenario. Lifting by itself is not sufficient for witnessing GMNL. 

Hence, in 
 \cite{Curchod_2019}, the authors extend the lifting technique to a construction of GMNL witnesses.
 Specifically, in a $K$-party Bell setup, they lift the bipartite CH inequalities between one party $A_1$ and each of the parties $A_k$ to $K$ parties. They sum the lifted inequalities to obtain a single one and, because this inequality could also be violated by biseparable correlations, they add correction terms that cancel any such violations out, in order to obtain a GMNL witness.

In the following we prove that \eqref{top_ineq} is  satisfied by all $K$-party biseparable non-signalling correlations and it thus a GMNL witness. 
Let $\mathcal{S} = \{\beta \subset \{1,...,K\}: 1\in\beta,\beta\neq \{1,...,K\}\}$, let $\beta\in\mathcal{S}$ and let $\bar \beta = \{1,...,K\}\setminus\beta$ be its complement. Let $P_{ext}$ be any extremal point of the biseparable non-signalling set $\mathcal{N}_{K}$. Then, we can find at least one bipartition $\beta|\bar \beta$ in the $K-$party Bell scenario such that
\begin{equation}
    P_{ext} = P_{\beta}(\vec a_{\beta}|\vec x_{\beta})P_{\bar \beta}(\vec a_{\bar \beta}| \vec x_{\bar \beta})
\end{equation}
where $\vec a =(a_i)_{i=1}^{K}$, $\vec x =(x_i)_{i=1}^{K}$, $\vec a_{\beta} = (a_i)_{i\in \beta}$, $\vec a_{\bar \beta} = (a_i)_{i\in \bar \beta}$, $\vec x_{\beta} = (x_i)_{i\in \beta}$ and $\vec x_{\bar \beta} = (x_i)_{i\in \bar \beta}$. Furthermore, let $j\in\bar \beta$, 
then we can rewrite \eqref{top_ineq} as
\begin{align}
    \nonumber I_K  \left( P_{ext} \right) &= I\left(Q^j_{P_{ext}}\right)  -  Q^j_{P_{ext}}(00|00) \nonumber \\
    &+  P_{ext}(0^{{ K-1}},0\ldots0|0^{{ K-1}},0\ldots 0) \nonumber\\&
    \nonumber+\sum_{\substack{k\in\bar \beta\\ k\neq j}} \left(I\left(Q^k_{P_{ext}}\right)-Q^k_{P_{ext}}(00|00)\right)  \\&
    + \!\!\!\! \sum_{k\in\beta\setminus \{1\}} \!\! \!\!\!\left(I\left(Q^k_{P_{ext}}\right)-Q^k_{P_{ext}}(00|00)\right) \label{main_ap_a}
\end{align}

In the following we will show that each summand in this expression is non-positive and hence that $I_K\left(P_{ext}\right) \leq 0$. 
Let us start with $I\left(Q^j_{P_{ext}}\right)$. Let $k\in \bar \beta$. Since $1\in \beta$, we can define $\vec a_{\beta \setminus \{1\}} = (a_i)_{i\in \beta \setminus \{1\}}$ and $\vec a_{\bar \beta \setminus \{k\}}=(a_i)_{i\in \bar \beta \setminus \{k\}}$ such that
\begin{equation}
    \vec a_{\beta} = (a_1,\vec a_{\beta \setminus \{1\}})\;\text{ and }\;
    \vec a_{\bar \beta} = (a_k,\vec a_{\bar \beta \setminus \{k\}}).
\end{equation}
Therefore, we can write
\begin{equation}
    P_{ext}(\vec a| \vec x) = P_\beta(a_1,\vec a_{\beta \setminus \{1\}}|\vec x_{\beta})P_{\bar \beta} (a_k,\vec a_{\bar \beta \setminus \{k\}}|\vec x_{\bar \beta}).
\end{equation}
Inserting this into \eqref{marginals} we can easily see that the lifted quantities factorize as
\begin{equation}\label{factorize}
    Q^k_{P_{ext}}(u,v|p,q) = R_k(u|p)S_k(v|q),
\end{equation}
where
\begin{align}\label{R}
&R_k(u|p)=\\&\sum_{\substack{a_1\in\{0,1\}^{K-1}\\:a_1^k=u}}P\left(a_1,\vec a_{\beta \setminus \{1\}}=0|x_1=pe_k,\vec x_{\beta \setminus \{1\}}=0\right)\nonumber
\end{align}
and
\begin{align}\label{S}
S_k(v|q)=P\left(a_k=v,\vec a_{\bar \beta \setminus \{k\}}=0|x_{k}=q,\vec x_{\bar \beta \setminus \{k\}}=0\right).
\end{align}
If $Q^k_{P_{ext}}$ was a probability distribution then $I(Q^k_{P_{ext}})\leq 0$ would hold in general as any distribution $P(u,v|p,q)=P(u|p)P(v|q)$ that is product across its inputs and outputs is by definition locally reproducible, i.e. in the local set $P\in \mathcal{L}$ and cannot violate the CH inequality $I(P)\leq 0$. However, in our construction $Q^k_{P_{ext}}$ is not a probability distribution therefore we will first show that this statement still holds. Looking at \eqref{R} we see that

\begin{align}
    &\sum_{u\in\{0,1\}}R_k(u|p) \\&= \sum_{a_1\in\{0,1\}^{K-1}}P_{\beta}\left( a_1,\vec a_{\beta \setminus \{1\}}=0|x_1=pe_k,\vec x_{\beta \setminus \{1\}}=0\right),\nonumber
\end{align}
and therefore,

\begin{align}
    \sum_{u\in\{0,1\}}R_k(u|p) = P_{\beta}\left( \vec a_{\beta \setminus \{1\}}=0|x_1=pe_k,\vec x_{\beta \setminus \{1\}}=0\right). 
\end{align}
Since $P_\beta$ is by definition a non-signalling distribution it satisfies
\begin{align}
    P_{\beta}\left( \vec a_{\beta \setminus \{1\}}=0|x_1=pe_k,\vec x_{\beta \setminus \{1\}}=0\right)=c^{\beta}_k \quad \forall p\in\{0,1\},
\end{align}
where $0\leq c^{\beta}_k \leq 1$ is a constant independent of $p$. Therefore,
\begin{align}
    \sum_{u\in\{0,1\}}R_k(u|p) = c_k^{\beta}
\end{align}
and 
\begin{equation}\label{ns1}
    \tilde R_k(u|p) := \frac{R_k(u|p)}{c_k^{\beta}}
\end{equation}
is a probability distribution as $\sum_{u}\tilde R_k(u|p)=1$ for all $p\in\{0,1\}$ and $0\leq \tilde R_k(u|p) \leq 1$. Similarly, from \eqref{S} we have
\begin{align}
    \sum_{v\in\{0,1\}}S_k(v|q)=P_{\bar \beta}\left(\vec a_{\bar \beta \setminus \{k\}}=0|x_k=q,\vec x_{\bar \beta \setminus \{k\}}=0\right)
\end{align}
and since $P_{\bar \beta}$ is also non-signalling by definition we have 
\begin{align}
    \sum_{v\in\{0,1\}}S_k(v|q) = c_k^{\bar \beta}
\end{align}
for some constant $0\leq c_k^{\bar \beta}\leq 1$ and thus,
\begin{equation}\label{ns2}
    \tilde S_k(v|q) := \frac{S_k(v|q)}{c_k^{\bar \beta}}
\end{equation}
is a probability distribution. Combining \eqref{ns1} and \eqref{ns2} with \eqref{factorize} we get
\begin{align}
    \tilde Q^k_{P_{ext}}(u,v|p,q)&:=\frac{Q^k_{P_{ext}}(u,v|p,q)}{c_k^{ \beta}c_k^{\bar \beta}} \\
    &=\tilde R_k(u|p)\tilde S_k(v|q),
\end{align}
which is a probability distribution. Therefore, by linearity of $I$,
\begin{equation}
I\left(Q^k_{P_{ext}}\right)=c_k^{ \beta}c_k^{\bar \beta}I\left(\tilde Q^k_{P_{ext}}\right).
\end{equation}
and since $\tilde Q^k_{P_{ext}}(u,v|p,q)$ is a probability distribution that factorizes as
\begin{equation}
\tilde Q^k_{P_{ext}}(u,v|p,q) = \tilde R_k(u|p)\tilde S_k(v|q),
\end{equation}
we have $I\left(\tilde Q^k_{P_{ext}}\right)\leq 0$ and thus 
\begin{equation}\label{first_third}
I\left( Q^k_{P_{ext}}\right)\leq 0\quad \forall k\in \bar\beta
\end{equation}
and since $j\in\bar \beta$ we get
\begin{equation}\label{first_sum}
    I\left( Q^j_{P_{ext}}\right)\leq 0
\end{equation}
for the first term in \eqref{main_ap_a}. 

For the second and third term in \eqref{main_ap_a} we have $j\in \bar \beta$ and let $\vec a_{\bar j} = (a_i)_{i\in\{1,...,K\}\setminus \{1,j\}}$. From \eqref{marginals} we have
\begin{align}
&Q^j_{P_{ext}}(00|00) =\\& \sum_{\substack{a_1\in \{0,1\}^{K-1}\\a_1^j=0}}P_{ext}(a_1,a_j=0,\vec a_{\bar j}=0|x_1=0,x_j=0,\vec x_{\bar j}=0).\nonumber
\end{align}
This sum contains the term $P_{ext}(0^{K-1},0,...,0|0^{K-1},0,...,0)$ and, therefore,
\begin{equation}\label{second_sum}
    P_{ext}(0^{K-1},0,...,0|0^{K-1},0,...,0)-Q^j_{P_{ext}}(00|00)\leq 0.
\end{equation}
For the fourth term, we know from \eqref{first_third} that
\begin{equation}
    \sum_{k\in \bar\beta}I\left( Q^k_{P_{ext}}\right)\leq 0
\end{equation}
Furthermore, it is easy to deduce from \eqref{marginals} that since $Q^k_{P_{ext}}(\cdot|\cdot)$ is a sum of probabilities we have $Q^k_{P_{ext}}(00|00)\geq 0$ for all $k\in\bar \beta$. Thus,
\begin{equation}\label{third_sum}
    \sum_{\substack{k\in\bar \beta\\ k\neq j}}\left(I\left(Q^k_{P_{ext}}\right)-Q^k_{P_{ext}}(00|00)\right)\leq 0.
\end{equation}
For the last term in \eqref{main_ap_a} let $k\in \beta \setminus \{1\}$. We rearrange the CH functional \eqref{CH_func} as
\begin{align}
    &I\left(Q^k_{P_{ext}}\right)-Q^k_{P_{ext}}(00|00)=\\&
    -Q^k_{P_{ext}}(01|01)-Q^k_{P_{ext}}(10|10)-Q^k_{P_{ext}}(00|11).\nonumber
\end{align}
Since, the lifted quantities are a summation of probabilities $Q^k_{P_{ext}}(\cdot|\cdot)\geq 0$ and thus,
\begin{equation}
    I\left(Q^k_{P_{ext}}\right)-Q^k_{P_{ext}}(00|00)\leq 0 \quad \forall k\in \beta \setminus \{1\}
\end{equation}
Thus, for this final term
\begin{equation}\label{fourth_sum}
    \sum_{k\in\beta\setminus \{1\}}\left(I\left(Q^k_{P_{ext}}\right)-Q^k_{P_{ext}}(00|00)\right)\leq 0.
\end{equation}
From \eqref{first_sum},\eqref{second_sum},\eqref{third_sum} and \eqref{fourth_sum} we see that all terms making up the expression in \eqref{main_ap_a} are non-positive and hence, $I_K\left(P_{ext}\right)\leq 0$ for all extremal biseparable points $P_{ext}$ of $\mathcal{N}_K$. Since the biseparable non-signalling set is a polytope for any distribution $P$ inside that polytope we also have $I_K\left(P\right)\leq 0$. \qed

\section{Proof of Proposition-\ref{prop1}} \label{app:prop1}

We first provide the following Lemma.

\begin{lem}\label{lem_B1} Let $P(a_1,\ldots,a_K|x_1,\ldots,x_K)$ be a $K$-partite distribution in the scenario of the main text, with $a_1=(a_1^2,\ldots,a_1^K)\in\{0,1\}^{K-1}$. Assume that there exists $j\in\{2,\ldots,K\}$ such that for all inputs $(x_1,\ldots,x_K)$, $P(a_j=1|x_1,\ldots,x_K)=1$ and $P(a_1^j=1|x_1,\ldots,x_K)=1$. Then the witness $I_K$ defined in \eqref{top_ineq} satisfies $I_K(P)=0$.
\end{lem}

\begin{proof}
  Recall the lifted bipartite quantities $Q_k^P$ from \eqref{marginals} for each $k\in\{2,\ldots,K\}$. We will first fix $k\neq j$ and show that $Q_k^P=0$ for all $k\neq j$. In every summand in \eqref{marginals}, $a_m=0$ holds for all $m\neq k$, in particular for $a_j=0$. Since $a_j=1$ deterministically by hypothesis, each such term is zero, hence $Q_k^P(u,v|p,q)=0$ for all $u,v,p,q$ and all $k\neq j$. Therefore, for any $k\neq j$, we have  $Q_k^P(00|00)=0$ and $I(Q_k^P)=0$.
  
  Next, we show that also $I(Q_j^P)=0$ and $Q_j^P(00|00)=0$. By hypothesis $a_1^j=1$ and $a_j=1$ deterministically, so $Q_j^P(u,v|p,q)=0$ whenever $u=0$ or $v=0$. In particular, 
\begin{equation}
Q_j^P(00|00)=Q_j^P(01|01)=Q_j^P(10|10)=Q_j^P(00|11)=0,
\end{equation}

which directly implies that $I(Q_j^P)=0$ and $Q_j^P(00|00)=0$. Thus, the two sums in \eqref{top_ineq} are zero. Furthermore, $P\left(0^{K-1},0,\ldots,0 \middle|\, 0^{K-1},0,\ldots,0\right)=0$ holds for $k\neq j$ and $k= j$, since $a_j =1$ deterministically. Hence, this term also equals zero concluding the proof of this lemma.   
\end{proof}

We will now provide another Lemma, which will be useful to identify the positive contributions to the GMNL witness.
\begin{lem}\label{lem_B2}Let $P(a_1,\ldots,a_K|x_1,\ldots,x_K)$ be a $K$-partite distribution and let $Q_k^P(u,v|p,q)$ be the lifted bipartite quantities defined in \eqref{marginals} for each $k\in\{2,\ldots,K\}$. Assume that for every $k\in\{2,\ldots,K\}$,
\[
Q_k^P(u,v|p,q)=c_k P_H(u,v|p,q),
\]
where $c_k$ is a constant and $P_H$ is a Hardy correlation satisfying
\begin{align}\nonumber
&P_H(00|00)=p(t)>0\\&P_H(01|01)=P_H(10|10)=P_H(00|11)=0. \label{eq:hardy}
\end{align}
Then the witness $I_K$ defined in \eqref{top_ineq} satisfies $I_K(P)\ge0$ and in particular for 
\begin{equation}\label{this_dist}  P(a_1,...,a_K|x_1,...,x_K)=\prod_{k=2}^KP_H(a_1^k,a_k|x_1^k,x_k)
\end{equation}
the witness in \eqref{top_ineq} satisfies $I_K(P)=(p(t))^{K-1}>0$.
\end{lem}

\begin{proof}
    By the assumption $Q_k^P=c_k P_H$ for all $k$,
\begin{align}\nonumber
I_K(P)&=\sum_{k=2}^K c_k I(P_H)+P(0^{K-1},0,\ldots,0|0^{K-1},0,\ldots,0)\\&-\sum_{k=2}^K c_k P_H(00|00).
\end{align}
Due to the properties \eqref{eq:hardy} of the Hardy correlations we have $I(P_H)=P_H(00|00)$. Hence the two sums cancel and
\begin{equation}
I_K(P)=P(0^{K-1},0,\ldots,0|0^{K-1},0,\ldots,0)\geq 0.
\end{equation}
For the distribution \eqref{this_dist}, 
\begin{align}
    &Q^k_P(u,v|p,q) =\\& \sum_{\substack{a_1\in\{0,1\}^{K-1}\\a_1^k=u}}\left[P_H(a^k_1,v|p,q)\prod^K_{\substack{m=2\\m\neq k}}P_H(a_1^m,0|0,0)\right]\nonumber
\end{align}
Now, $a^k_1=u$ is fixed by the summation constraint therefore defining $\vec a^*_1=(a_1^m)_{m\in\{2,...,K\}\setminus \{k\}}$ we can rewrite this as
\begin{align}
    Q^k_P(u,v|p,q) &=P_H(u,v|p,q) \! \! \! \! \! \!\sum_{\vec a_1^*\in\{0,1\}^{K-2}}\prod^K_{\substack{m=2\\m\neq k}}P_H(a_1^m,0|0,0)\nonumber \\
    &= c_k P_H(u,v|p,q)
\end{align}
for the constant 
\begin{align}
    c_k :&= \sum_{\vec a_1^*\in\{0,1\}^{K-2}}\prod^K_{\substack{m=2\\m\neq k}}P_H(a_1^m,0|0,0) \\
    &= \quad \prod^K_{\substack{m=2\\m\neq k}}\sum_{a_1^m\in\{0,1\}}P_H(a_1^m,0|0,0)\\
    &= \quad \prod^K_{\substack{m=2\\m\neq k}}\left[P_H(0,0|0,0)+P_H(1,0|0,0)\right] \\
    &= \quad \left[P_H(00|00)+P_H(10|00)\right]^{K-2}.
\end{align}

For the distribution in \eqref{this_dist},furthermore
\begin{align}
P(0^{K-1},0,\ldots,0|0^{K-1},0,\ldots,0) &=  \prod_{k=2}^KP_H(00|00) \nonumber \\
&=(p(t))^{K-1}, \nonumber
\end{align}
and $(p(t))^{K-1}>0$ concluding the proof for this lemma.
\end{proof}

We now have all ingredients to present the proof of Proposition-\ref{prop1}. For this, let us consider the following decomposition of \eqref{hardy_fam}:
\begin{equation}
    P=\frac{1}{K-1}\sum_{k=2}^{K}P^{(k)},
\end{equation}
where
\begin{align}\nonumber
P^{(k)}(a_1,\ldots,a_K|x_1,\ldots&,x_K)=P_H(a_1^k,a_k|x_1^k,x_k)\\
&\prod_{m\in\{2,\ldots,K\}\setminus\{k\}}\delta_{1,a_1^{m}}\delta_{1,a_m}.
\end{align}
for all $k\in\{2,\ldots,K\}$.

Using this decomposition for each of the $K-1$ copies of $P$ and multi-linearity of the AND wiring map $\mathcal{F}$, we can expand
\begin{align}\nonumber
P'&=\frac{1}{(K-1)^{K-1}}\times \\&\sum_{(k_1,...,k_{K-1})\in \{2,...,K\}^{K-1}}\mathcal{F}^{\left(\Vec{k}\right)}\!\left[P^{(k_1)},...,P^{(k_{K-1})}\right].\label{wired_dist_appendA}
\end{align}
where $\Vec{k} = (k_1,...,k_{K-1})$. We then define the sets
\begin{align}\nonumber
    &\mathcal{S} = \left\{\Vec{k}\in\{2,...,K\}^{K-1}:\{k_1,...,k_{K-1}\} = \{2,...,K\}\right\}\\&
    \mathcal{D} = \{2,...,K\}^{K-1}\setminus\mathcal{S}\nonumber
\end{align}
where $\mathcal{S}$ is the set of $\Vec{k}$ that are permutations of $\{2,...,K\}$ and $\mathcal{D}$ is its complement. For $\Vec{k}\in\mathcal{D}$ there exists at least one index $j\in\{2,...,K\}$ which does not appear among $k_1,...,k_{K-1}$. For this index $j$ we always have $a_1^j=1$ and $a_j=1$ for all inputs as determined by $\delta_{1,a_{1}^{j}} \delta_{1,a_j}$ and after applying the AND-wiring across all copies we still have deterministically $a_1^j=1$ and $a_j=1$. Thus,  by Lemma-\ref{lem_B1}, the GMNL witness for these terms vanishes
\begin{align}\label{IkD}
    I_K\left(\mathcal{F}^{(\Vec{k})}\left[P^{(k_1)},...,P^{(k_{K-1})}\right]\right) = 0,\;\;\;\forall \Vec{k}\in\mathcal{D}.
\end{align}

On the other hand, let $\Vec{k}\in\mathcal{S}$. Then for each $m\in\{2,...,K\}$ there exists exactly one copy $\ell$ such that $k_\ell = m$ that contributes $P_H$ while for all other copies the parties deterministically output 1 and hence, AND-wiring across copies leaves the Hardy correlations unchanged. Therefore we get
 \begin{align}
     \mathcal{F}^{(\Vec{k})}\left[P^{(k_1)},...,P^{(k_{K-1})}\right] = \prod_{k=2}^K P_H(a_1^k,a_k|x_1^k,x_k).
 \end{align}
Invoking Lemma-\ref{lem_B2}, we obtain
\begin{align}\label{IkS}
    I_K\left(\mathcal{F}^{(\Vec{k})}\left[P^{(k_1)},...,P^{(k_{K-1})}\right]\right) = p^{K-1}(t),\;\;\;\forall \Vec{k}\in\mathcal{S}.
\end{align}
There are exactly $(K-1)!$ vectors  $\Vec{k}\in\mathcal{S}$, i.e., $|\mathcal{S}|=(K-1)!$. Therefore, combining \eqref{IkD} and \eqref{IkS} with \eqref{wired_dist_appendA} we get
\begin{align}\nonumber
    I_K(P')&=\frac{|\mathcal{S}|}{(K-1)^{K-1}}p(t)^{K-1}\\&=\frac{(K-1)!}{(K-1)^{K-1}}p(t)^{K-1}>0,\label{r1ap}
\end{align}

concluding the proof for Proposition-\ref{prop1}.

\section{Proof of Proposition-\ref{prop2}}
\label{app:catalyst}
The initial distribution of system and catalyst is 
\begin{align}
    &P_{\mathrm{init}}\left(\mathbf{a}^{[0]},...,\mathbf{a}^{[n-1]},\mathbf{\tilde a}\big| \mathbf{x}^{[0]},...,\mathbf{x}^{[n-1]}\right)\\& = P\left(\mathbf{a}^{[0]}\big|\mathbf{x}^{[0]}\right) P_c\left( \mathbf{a}^{[1]},...,\mathbf{a}^{[n-1]},\mathbf{\tilde a} \big| \mathbf{x}^{[1]},..., \mathbf{x}^{[n-1]}\right), 
\end{align}
where we label the system distribution $P$ by $[0]$, and the $(n-1)$ catalyst distributions by $[1],\ldots,[n-1]$. For each distribution $[s]$, we express the local outputs as $\mathbf{a}^{[s]}=(a_1^{[s]},\ldots,a_K^{[s]})$ such that $a_k^{[s]}\in \mathcal{A}_k$, and similarly for inputs $\mathbf{x}^{[s]}=(x_1^{[s]},\ldots,x_K^{[s]})$ such that $x_k^{[s]}\in \mathcal{X}_k$. The outputs of the classical catalyst register are labelled $\mathbf{\tilde a}=(\tilde a_1,\ldots,\tilde a_K)$ with $\tilde a_k\in\{0,\ldots,n-1\}$, such that we have $\tilde a_1=\cdots=\tilde a_K=i$ deterministically.

We can further rewrite 
\begin{align}
&P_c=\frac1n\sum_{i=0}^{n-1}C_i\left( \mathbf{a}^{[1]},...,\mathbf{a}^{[n-1]},\mathbf{\tilde a} \big| \mathbf{x}^{[1]},..., \mathbf{x}^{[n-1]}\right),  
\end{align}
where for all $i\in\{1,...,n-2\}$,
\begin{align}\nonumber
&C_i\left( \mathbf{a}^{[1]},...,\mathbf{a}^{[n-1]},\mathbf{\tilde a} \big| \mathbf{x}^{[1]},..., \mathbf{x}^{[n-1]}\right)\\&=\prod_{\ell=1}^{i}P\left(\mathbf{a}^{[\ell]}|\mathbf{x}^{[\ell]}\right)\ \prod_{\ell=i+1}^{n-1} P_*\left(\mathbf{a}^{[\ell]}|\mathbf{x}^{[\ell]}\right)\prod_{k=1}^K \delta_{i,\tilde a_k},
\end{align}
while for $i=0$ and $i=n-1$ we have 
\begin{align}\nonumber
C_0&\left( \mathbf{a}^{[1]},...,\mathbf{a}^{[n-1]},\mathbf{\tilde a} \big| \mathbf{x}^{[1]},..., \mathbf{x}^{[n-1]}\right)\\&\quad \quad \quad\quad \quad \quad\quad \quad \quad=\prod_{\ell=1}^{n-1}P_*\left(\mathbf{a}^{[\ell]}|\mathbf{x}^{[\ell]}\right)\prod_{k=1}^K \delta_{0,\tilde a_k}
\end{align}
and 
\begin{align}\nonumber
&C_{n-1}\left( \mathbf{a}^{[1]},...,\mathbf{a}^{[n-1]},\mathbf{\tilde a} \big| \mathbf{x}^{[1]},..., \mathbf{x}^{[n-1]}\right)\\&\quad \quad \quad\quad \quad \quad\quad \quad \quad=\prod_{\ell=1}^{n-1}P\left(\mathbf{a}^{[\ell]}|\mathbf{x}^{[\ell]}\right)\prod_{k=1}^K \delta_{n-1,\tilde a_k}
\end{align}
respectively.

Let us consider a classical local processing $\mathcal{G}$ that creates a larger total box defined as
\begin{align}    \nonumber&\mathcal{G}\left(P_{\mathrm{init}}\left(\mathbf{a}^{[0]},...,\mathbf{a}^{[n-1]},\mathbf{\tilde a}\big| \mathbf{x}^{[0]},...,\mathbf{x}^{[n-1]}\right)\right)\\&
    =P_{\mathrm{tot}}\left(\mathbf{a}_{\mathrm{tar}}, \mathbf{a}^{[1]}_{\mathrm{cat}},...,\mathbf{a}^{[n-1]}_{\mathrm{cat}},\mathbf{\tilde a}' \big|\mathbf{x}_{\mathrm{tar}}, \mathbf{x}^{[1]}_{\mathrm{cat}},...,\mathbf{x}^{[n-1]}_{\mathrm{cat}} \right)
\end{align}
where $\mathbf{a}_{\mathrm{tar}} = (a_{1[\mathrm{tar}]},...,a_{K[\mathrm{tar}]})$ and $\mathbf{x}_{\mathrm{tar}} = (x_{1[\mathrm{tar}]},...,x_{K[\mathrm{tar}]})$ and for all $s\in\{1,...,n-1\}$:

\begin{equation}
\mathbf{a}^{[s]}_{\mathrm{cat}} = (a^{[s]}_{1[\mathrm{cat}]},...,a^{[s]}_{K[\mathrm{cat}]}),
\quad
\mathbf{x}^{[s]}_{\mathrm{cat}} = (x^{[s]}_{1[\mathrm{cat}]},...,x^{[s]}_{K[\mathrm{cat}]}),
\end{equation}
such that $a_{k[\mathrm{tar}]},a^{[s]}_{k[\mathrm{cat}]}\in \mathcal{A}_k$ and $x_{k[\mathrm{tar}]},x^{[s]}_{k[\mathrm{cat}]}\in \mathcal{X}_k$ and $\mathbf{\tilde a}'=(\tilde a_{1}',...,\tilde a_{K}')$ such that $\tilde a_{k}'\in\{0,...,n-1\}$. 

Now, to construct $\mathcal{G}$, let each party read their local value $\tilde a_k = i$. Notice that the parties either observe $\tilde a_1=\cdots=\tilde a_K=i\in\{0,\ldots,n-2\}$ or $\tilde a_1=\cdots=\tilde a_K=i=n-1$. 

If the parties observe $\tilde a_1=\cdots=\tilde a_K=i\in\{0,\ldots,n-2\}$ they know that over their $n$ distributions there are $i+1<n$ copies of $P$. Therefore, they set the target outputs deterministically as
\begin{align}
\mathbf{a}_{\mathrm{tar}}=\left(f_1\left(x_{1[\mathrm{tar}]}\right),...,f_K\left(x_{K[\mathrm{tar}]}\right)\right),
\end{align}
where $f_k:\mathcal{X}_i \rightarrow \mathcal{A}_i$ are local response functions such that $P_*(a_1,...,a_K|x_1,...,x_K)=\prod_{k=1}^{K}\delta_{a_k,f_k(x_k)}$. Hence, in this case the target inputs $\mathbf{x}_{\mathrm{tar}}$ and outputs $\mathbf{a}_{\mathrm{tar}}$ follow the distribution $P_*$. For the catalyst inputs and outputs they follow a strategy to partially recover the catalyst as they set
\begin{align}\nonumber
&\mathbf{x}^{[0]}=\mathbf{x}_{\mathrm{cat}}^{[i+1]},\mathbf{a}_{\mathrm{cat}}^{[i+1]}=\mathbf{a}^{[0]},\\\nonumber& \mathbf{x}^{[s]}=\mathbf{x}_{\mathrm{cat}}^{[s]}\text{ and }\mathbf{a}_{\mathrm{cat}}^{[s]}=\mathbf{a}^{[s]}\ \ \forall s\in\{1,\ldots,n-1\}\setminus\{i+1\} \\&\tilde a_k'=i+1\ \ \forall k.
\label{B1}
\end{align}

To see how this strategy works let $C_i'\left( \mathbf{a}^{[1]}_{\mathrm{cat}},...,\mathbf{a}^{[n-1]}_{\mathrm{cat}},\mathbf{\tilde a}' \big| \mathbf{x}^{[1]}_{\mathrm{cat}},...,\mathbf{x}^{[n-1]}_{\mathrm{cat}}\right)$ denote the returned-catalyst produced when the input catalyst is $C_i\left( \mathbf{a}^{[1]},...,\mathbf{a}^{[n-1]},\mathbf{\tilde a} \big| \mathbf{x}^{[1]},..., \mathbf{x}^{[n-1]}\right)$. Then under $C_i$ the catalyst distributions $[1],\ldots,[i]$ are distributed as $P$ and the distributions $[i+1],\ldots,[n-1]$ are distributed as $P_*$. The rule \eqref{B1} replaces the $i+1$-th distribution by the $0$-th, which is $P$, and leaves all others unchanged. If further sets $\tilde a'=i+1$, hence
\begin{align}\nonumber
&C_i'\left( \mathbf{a}^{[1]}_{\mathrm{cat}},...,\mathbf{a}^{[n-1]}_{\mathrm{cat}},\mathbf{\tilde a}' \big| \mathbf{x}^{[1]}_{\mathrm{cat}},...,\mathbf{x}^{[n-1]}_{\mathrm{cat}}\right)\\&=C_{i+1}\left( \mathbf{a}^{[1]}_{\mathrm{cat}},...,\mathbf{a}^{[n-1]}_{\mathrm{cat}},\mathbf{\tilde a}' \big| \mathbf{x}^{[1]}_{\mathrm{cat}},...,\mathbf{x}^{[n-1]}_{\mathrm{cat}}\right).
\end{align}

Therefore, in this case the parties obtain the distribution

\begin{align}\nonumber
    &P_{i<n-1}\left(\mathbf{a}_{\mathrm{tar}}, \mathbf{a}^{[1]}_{\mathrm{cat}},...,\mathbf{a}^{[n-1]}_{\mathrm{cat}},\mathbf{\tilde a}' \big|\mathbf{x}_{\mathrm{tar}}, \mathbf{x}^{[1]}_{\mathrm{cat}},...,\mathbf{x}^{[n-1]}_{\mathrm{cat}} \right)
    \\\nonumber &= P_{*}\left(\mathbf{a}_{\mathrm{tar}}|\mathbf{x}_{\mathrm{tar}}\right) C_{i+1}\left(\mathbf{a}^{[1]}_{\mathrm{cat}},...,\mathbf{a}^{[n-1]}_{\mathrm{cat}},\mathbf{\tilde a}' \big| \mathbf{x}^{[1]}_{\mathrm{cat}},...,\mathbf{x}^{[n-1]}_{\mathrm{cat}}\right)
\end{align}

Now, let's look at the case where the parties observe $\tilde a_1=\cdots=\tilde a_K=i=n-1$. In this case, they know that over their $n$ distributions they hold exactly $n$ copies of the box $P$ that they can wire with $\mathcal{W}$. Therefore, they set $\mathbf{a}_{\mathrm{tar}}$ using the wiring $\mathcal{W}$ as

\begin{align}\nonumber
    &\mathcal{W}\left[P\left(\mathbf{a}^{[0]}|\mathbf{x}^{[0]}\right),P\left(\mathbf{a}^{[1]}|\mathbf{x}^{[1]}\right),...,P\left(\mathbf{a}^{[n-1]}|\mathbf{x}^{[n-1]}\right)\right] \\&= P'\left(\mathbf{a}_{\mathrm{tar}}|\mathbf{x}_{\mathrm{tar}}\right).
\end{align}
Hence, in this case the target inputs $\mathbf{x}_{\mathrm{tar}}$ and outputs $\mathbf{a}_{\mathrm{tar}}$ follow the distribution $P'$. For the catalyst, they set the outputs deterministically as

\begin{align}\nonumber
    \mathbf{a}_{\mathrm{cat}}^{\mathrm[s]} &= \left(f_1\left(x_{1[\mathrm{cat}]}^{[s]}\right),...,f_K\left(x_{K[\mathrm{cat}]}^{[s]}\right)\right)\quad \forall s\in\{1,...,n-1\},\\
    \tilde a_k'&=0\qquad \forall k. \label{B3}
\end{align}
Notice that this rule outputs $C_0$, hence
\begin{align}\nonumber
&C_{n-1}'\left( \mathbf{a}^{[1]}_{\mathrm{cat}},...,\mathbf{a}^{[n-1]}_{\mathrm{cat}},\mathbf{\tilde a}' \big| \mathbf{x}^{[1]}_{\mathrm{cat}},...,\mathbf{x}^{[n-1]}_{\mathrm{cat}}\right)\\&=C_{0}\left( \mathbf{a}^{[1]}_{\mathrm{cat}},...,\mathbf{a}^{[n-1]}_{\mathrm{cat}},\mathbf{\tilde a}' \big| \mathbf{x}^{[1]}_{\mathrm{cat}},...,\mathbf{x}^{[n-1]}_{\mathrm{cat}}\right).
\end{align}
Therefore, in this case they obtain the distribution
\begin{align}
    \nonumber&
    P_{i=n-1}\left(\mathbf{a}_{\mathrm{tar}}, \mathbf{a}^{[1]}_{\mathrm{cat}},...,\mathbf{a}^{[n-1]}_{\mathrm{cat}},\mathbf{\tilde a}' \big|\mathbf{x}_{\mathrm{tar}}, \mathbf{x}^{[1]}_{\mathrm{cat}},...,\mathbf{x}^{[n-1]}_{\mathrm{cat}} \right)\\&
    = P'\left(\mathbf{a}_{\mathrm{tar}}|\mathbf{x}_{\mathrm{tar}}\right) C_0\left(\mathbf{a}^{[1]}_{\mathrm{cat}},...,\mathbf{a}^{[n-1]}_{\mathrm{cat}},\mathbf{\tilde a}' \big| \mathbf{x}^{[1]}_{\mathrm{cat}},...,\mathbf{x}^{[n-1]}_{\mathrm{cat}}\right).
\end{align}
Since $P_c(\tilde a_1=\cdots=\tilde a_K=i)=1/n$ for all $i$, the total distribution becomes
\begin{widetext}
\begin{align}\nonumber
    P_{\mathrm{tot}} \left(\mathbf{a}_{\mathrm{tar}}, \mathbf{a}^{[1]}_{\mathrm{cat}},...,\mathbf{a}^{[n-1]}_{\mathrm{cat}},\mathbf{\tilde a}' \big|\mathbf{x}_{\mathrm{tar}}, \mathbf{x}^{[1]}_{\mathrm{cat}},...,\mathbf{x}^{[n-1]}_{\mathrm{cat}} \right)  
\nonumber
=& \frac{1}{n}P_*\left(\mathbf{a}_{\mathrm{tar}}|\mathbf{x}_{\mathrm{tar}}\right)\sum^{n-1}_{i=1}C_{i}\left( \mathbf{a}^{[1]}_{\mathrm{cat}},...,\mathbf{a}^{[n-1]}_{\mathrm{cat}},\mathbf{\tilde a}' \big| \mathbf{x}^{[1]}_{\mathrm{cat}},...,\mathbf{x}^{[n-1]}_{\mathrm{cat}}\right) \\
+&\frac{1}{n}P'\left(\mathbf{a}_{\mathrm{tar}}|\mathbf{x}_{\mathrm{tar}}\right)C_0\left( \mathbf{a}^{[1]}_{\mathrm{cat}},...,\mathbf{a}^{[n-1]}_{\mathrm{cat}},\mathbf{\tilde a}' \big| \mathbf{x}^{[1]}_{\mathrm{cat}},...,\mathbf{x}^{[n-1]}_{\mathrm{cat}}\right)
\end{align}
\end{widetext}

Taking the marginal of this distribution over $\mathbf{a}_{\mathrm{tar}}$ we get
\begin{align}\nonumber
    &P_{\mathrm{cat}}\left( \mathbf{a}^{[1]}_{\mathrm{cat}},...,\mathbf{a}^{[n-1]}_{\mathrm{cat}} \big| \mathbf{x}^{[1]}_{\mathrm{cat}},...,\mathbf{x}^{[n-1]}_{\mathrm{cat}}\right)\\& = \frac{1}{n}\sum_{i=0}^{n-1}C_i\left( \mathbf{a}^{[1]}_{\mathrm{cat}},...,\mathbf{a}^{[n-1]}_{\mathrm{cat}} \big| \mathbf{x}^{[1]}_{\mathrm{cat}},...,\mathbf{x}^{[n-1]}_{\mathrm{cat}}\right)
\end{align}
and hence the catalyst is recovered as $P_{\mathrm{cat}} = P_c$. Similarly, taking the marginal of this distribution with respect to $\mathbf{a}_{\mathrm{cat}}$ we get
\begin{equation}
    P_{\mathrm{tar}}\left(\mathbf{a}_{\mathrm{tar}}|\mathbf{x}_{\mathrm{tar}}\right) = \frac{1}{n}P'\left(\mathbf{a}_{\mathrm{tar}}|\mathbf{x}_{\mathrm{tar}}\right) + \frac{n-1}{n}P_*\left(\mathbf{a}_{\mathrm{tar}}|\mathbf{x}_{\mathrm{tar}}\right).
\end{equation}

Hence, $I^*(P_{\mathrm{tar}})=\frac{1}{n}I^*(P^\prime)$. Since $I^*(P^\prime)>0$ is given a priori we also have $I^*(P^{\prime}_{\mathrm{tar}})>0$ concluding the proof of Proposition-\ref{prop2}. Furthermore, if we choose the original distribution to be \eqref{hardy_fam} and the GMNL witness to be $I^*(P) = I_3(P)$, combining this with \eqref{r1ap} from Proposition-\ref{prop1} we get

\begin{equation}
    I_K(P^{\prime}_{\mathrm{tar}}) = \frac{n!}{n^{K}}p(t)^{n}>0.
\end{equation}

\section{Example of GMNL activation via catalysis for the distribution in Proposition-\ref{prop1}}
\label{app:new_wiring}

 We choose the original distribution to be \eqref{ex1} and the GMNL witness to be $I^*(P) = I_3(P)$ as given in \eqref{top_ineq}. Furthermore, we set the local deterministic distribution $P_*=P_0$ such that it outputs $1$ for any given combination of inputs. In this case, the original distribution between the parties is given by $P(a_1,a_2,a_3|x_1,x_2,x_3)$ and the catalyst is given by

\begin{align}\nonumber
    &P_c(a^c_1,a^c_2,a^c_3,\Tilde{a}_1,\Tilde{a}_2,\Tilde{a}_3|x^c_1,x^c_2,x^c_3) =\\ \nonumber &\frac{1}{2}P(a^c_1,a^c_2,a^c_3|x^c_1,x^c_2,x^c_3)\delta_{1,\Tilde{a}_1}\delta_{1,\Tilde{a}_2}\delta_{1,\Tilde{a}_3}\\&
    + \frac{1}{2}P_0(a^c_1,a^c_2,a^c_3|x^c_1,x^c_2,x^c_3)\delta_{0,\Tilde{a}_1}\delta_{0,\Tilde{a}_2}\delta_{0,\Tilde{a}_3}.
\end{align}
where the superscript $c$ denotes outputs of the catalyst and $(\tilde a_1,\tilde a_2, \tilde a_3)$ are outputs of the classical catalyst registers. Given $P$ and $P_c$, the parties perform classical and local operations to construct a new box 
\begin{equation}\nonumber
    P_{\mathrm{tot}}^{ \prime}(a^{\prime}_1,a^{\prime}_2,a^{\prime}_3,a^{c^\prime}_1,a^{c^\prime}_2,a^{c^\prime}_3,\Tilde{a}^\prime_1,\Tilde{a}^\prime_2,\Tilde{a}^\prime_3|x^{\prime}_1,x^{\prime}_2,x^{\prime}_3,x^{c^\prime}_1,x^{c^\prime}_2,x^{c^\prime}_3)
\end{equation}
such that the marginals of this distribution are given by
\begin{align}\nonumber
&P_{\mathrm{tar}}^{\prime}(a^{\prime}_1,a^{\prime}_2,a^{\prime}_3|x^{\prime}_1,x^{\prime}_2,x^{\prime}_3)=\sum_{\substack{a^{c^\prime}_1,a^{c^\prime}_2,a^{c^\prime}_3\\\Tilde{a}^\prime_1,\Tilde{a}^\prime_2,\Tilde{a}^\prime_3}} P_{\text{tot}}^{\prime}(...)
\\&P_c^{\prime}(a^{c^\prime}_1,a^{c^\prime}_2,a^{c^\prime}_3,\Tilde{a}^\prime_1,\Tilde{a}^\prime_2,\Tilde{a}^\prime_3|x^{c^\prime}_1,x^{c^\prime}_2,x^{c^\prime}_3)= \sum_{a^{\prime}_1,a^{\prime}_2,a^{\prime}_3}P_{\text{tot}}^{\prime}(...)
\end{align}

If $P^{\prime}_c$ is the same as the initial catalyst $P_c$ and if the target distribution $P^{\prime}_{\mathrm{tar}}$ is GMNL 
then we have GMNL superactivation for this example. These classical, local operations that the parties perform are the following:
\begin{widetext}
\centering
\begin{tabular}{|c|c|c|c|c|}
\hline
Party & Input registers & Target registers & Catalyst registers & $\begin{array}{c}
     \text{Classical catalyst} \\
     \text{registers} 
\end{array}$ \\
\hline
$A_2,A_3$ &
$\begin{array}{c}\\[-10pt]
     x_k=(\tilde a_k\oplus1)x_k^{c\prime} \oplus \tilde a_kx'_k \\[2pt]
     x^c_k= x'_k\\[1.5pt]
\end{array}$
&
$\begin{array}{c}
a_k'=(\tilde a_k\oplus1)a_k^c \oplus \tilde a_k a_ka_k^c 
\end{array}$ &
$\begin{array}{c}
a_k^{c\prime}=(\tilde a_k\oplus1)a_k \oplus \tilde a_k 
\end{array}$ &
$\tilde a_k'=\tilde a_k\oplus1$ \\
\hline
$A_1$ &
$\begin{array}{c}\\[-10pt]
     x^k_1=(\tilde a_1\oplus1)x_1^{c\prime k} \oplus \tilde a_1 x^{\prime k}_1 \\[2pt]
     x^{ck}_1= x^{\prime k}_1\\[1.5pt]
\end{array}$
&
$\begin{array}{c}
a_1^{\prime k}=(\tilde a_1\oplus1)a_1^{ck}\oplus \tilde a_1a_1^ka_1^{ck} 
\end{array}$ &
$\begin{array}{c}
a_1^{c\prime k}=(\tilde a_1\oplus1)a_1^k\oplus \tilde a_1 
\end{array}$ &
$\tilde a_1'=\tilde a_1\oplus1$ \\
\hline
\end{tabular}
\end{widetext}

 Intuitively, if the $k$-th party observes $\Tilde{a}_k=1$, they implement the following functionality (i): They use the inputs of the target box $x^\prime_k$ as the inputs of the initial box $x_k$ and of the catalyst $x^c_k$. They then add a post-processing that applies a classical AND-operation to the outputs of their initial box $a_k$ and of the catalyst $a^c_k$, which generates the target output $a^\prime_k$, (ii): they deterministically set the output $a^{c^\prime}_k=1$ and they prepare a classical register $\Tilde{a}^\prime_k = 0$.  If the parties observe $\Tilde{a}_k=0$ instead (i): they use the target inputs $x_k^\prime$ as inputs of the catalyst $x^c_k$ and set the outputs of the target box $a^\prime_k$ to be the outputs of the catalyst $a^c_k$, (ii): they use the inputs $x^{c^\prime}_k$ as inputs of the initial box $x_k$ and set the outputs of  $a^{c^\prime}_k$ to be the outputs of their initial box, $a_k$. Therefore, with probability $P(\tilde a_1 = \tilde a_2= \tilde a_3=1)=1/2$, they have created the new box with marginal target behaviour $P^{\prime}(a^{\prime}_1,a^{\prime}_2,a^{\prime}_3|x^{\prime}_1,x^{\prime}_2,x^{\prime}_3)$ from Proposition-\ref{prop1}, while the catalyst part follows the distribution $P_0(a^{c^\prime}_1,a^{c^\prime}_2,a^{c^\prime}_3|x^{c^\prime}_1,x^{c^\prime}_2,x^{c^\prime}_3)$ and $\delta_{0,\Tilde{a}_1}\delta_{0,\Tilde{a}_2}\delta_{0,\Tilde{a}_3}$. Similarly, with probability $P(\tilde a_1 = \tilde a_2= \tilde a_3=0)=1/2$, the target box has a distribution $P_0(a^{\prime}_1,a^{\prime}_2,a^{\prime}_3|x^{\prime}_1,x^{\prime}_2,x^{\prime}_3)$ and the catalyst's outputs follow the distribution $P(a^{c^\prime}_1,a^{c^\prime}_2,a^{c^\prime}_3|x^{c^\prime}_1,x^{c^\prime}_2,x^{c^\prime}_3)$ and $\delta_{1,\Tilde{a}_1}\delta_{1,\Tilde{a}_2}\delta_{1,\Tilde{a}_3}$. Thus, overall, the marginal on the catalyst part of the box is
\begin{align}\nonumber
    &P_c^{\prime}(a^{c^\prime}_1,a^{c^\prime}_2,a^{c^\prime}_3,\Tilde{a}^\prime_1,\Tilde{a}^\prime_2,\Tilde{a}^\prime_3|x^{c^\prime}_1,x^{c^\prime}_2,x^{c^\prime}_3)=\\&\nonumber\frac{1}{2}P_0(a^{c^\prime}_1,a^{c^\prime}_2,a^{c^\prime}_3|x^{c^\prime}_1,x^{c^\prime}_2,x^{c^\prime}_3)\delta_{0,\Tilde{a}_1}\delta_{0,\Tilde{a}_2}\delta_{0,\Tilde{a}_3}\\&+
    \frac{1}{2}P(a^{c^\prime}_1,a^{c^\prime}_2,a^{c^\prime}_3|x^{c^\prime}_1,x^{c^\prime}_2,x^{c^\prime}_3)\delta_{1,\Tilde{a}_1}\delta_{1,\Tilde{a}_2}\delta_{1,\Tilde{a}_3}
\end{align}
which is the same as the original catalyst, meaning that the catalysis requirement is satisfied. In addition, the resulting marginal on the target system is
\begin{align}\nonumber
    P^{\prime }_{\mathrm{tar}}(a^{\prime}_1,a^{\prime}_2,a^{\prime}_3|x^{\prime}_1,x^{\prime}_2,x^{\prime}_3)&=\frac{1}{2}P^{\prime}(a^{\prime}_1,a^{\prime}_2,a^{\prime}_3|x^{\prime}_1,x^{\prime}_2,x^{\prime}_3)\\&+\frac{1}{2}P_0(a^{\prime}_1,a^{\prime}_2,a^{\prime}_3|x^{\prime}_1,x^{\prime}_2,x^{\prime}_3).
\end{align}
We know from Proposition-\ref{prop1} that $I_3(P^\prime)=\frac{1}{2}p^2(t)>0$ and it is trivial to check that the deterministic distribution satisfies $I_3(P_0) = 0$. Hence, the GMNL witness for the target distribution is strictly positive as $I_3(P^{\prime}_{\mathrm{tar}})=\frac{1}{4}p^2(t)>0$. 

\end{document}